\documentclass{article}
\usepackage{authblk}
\usepackage{amsthm}
\usepackage{enumerate}
\usepackage{mathptmx}
\usepackage{latexsym}
\usepackage{graphicx}
\usepackage{float}
\usepackage{subfigure}
\usepackage[linesnumbered,ruled]{algorithm2e}
\usepackage{amssymb}
\usepackage{amsmath}
\usepackage{marvosym}
\usepackage{caption}
\usepackage{color,soul}
\setstcolor{red}
\usepackage{suffix}
\usepackage{epstopdf}

\renewcommand{\phi}{\varphi}
\newcommand{\myqed}{}

\newtheorem{lm}{Lemma}

\newtheorem{t1}{Theorem}

\DeclareMathOperator*{\wt}{wt}
\DeclareMathOperator*{\wtphi}{wt_{\varphi}}

\DeclareMathOperator*{\supp}{supp}
\DeclareMathOperator*{\br}{Br}

\newcommand{\ineff}[4]{\Delta_#1(#2,#3;#4)}

\WithSuffix\newcommand\ineff*[3]{\Delta_#1(#2;#3)}

\newcommand{\ief}[3]{\Delta(#1,#2;#3)}
\WithSuffix\newcommand\ief*[2]{\Delta(#1;#2)}

\newcommand{\diste}{\delta}
\newcommand{\distephi}{\delta_{\phi}}
\newcommand{\dist}{d }
\newcommand{\distphi}{d_{\phi}}
\newcommand{\Disp}{D}
\newcommand{\Dispphi}{D_{\phi}}
\newcommand{\cv}{v_c} % central vertex
\newcommand{\centerineff}{C}
\newcommand{\keywords}[1]{\par\noindent\textbf{Keywords: }#1}
\begin{document}
\title{Computing Similarity Distances Between Rankings\thanks{This work was supported in part by the NSF STC Class 2010 CCF 0939370 grant. Research of the third author is supported by the IC Postdoctoral Research Fellowship. Farzad Farnoud is now at the Department of Electrical Engineering at the California Institute of Technology.}}
\author[1]{Farzad Farnoud}
\author[1]{Lili Su}
\author[2]{Gregory J.~Puleo}
\author[1]{Olgica Milenkovic}
\affil[1]{Department of Electrical and Computer Engineering\\ University of Illinois at Urbana-Champaign\\ Urbana, IL, USA}
\affil[2]{Coordinated Science Lab\\ University of Illinois at Urbana-Champaign\\ Urbana, IL, USA}
\maketitle

%\fcomment{TO DO: fix inconsistency between transform and decomposition}
\begin{abstract}
  We address the problem of computing distances between rankings that
  take into account similarities between candidates. The need for
  evaluating such distances is governed by applications as diverse as
  rank aggregation, bioinformatics, social sciences and data
  storage. The problem may be summarized as follows: Given two
  rankings and a positive cost function on transpositions that depends
  on the similarity of the candidates involved, find a smallest cost
  sequence of transpositions that converts one ranking into
  another. Our focus is on costs that may be described via special
  metric-tree structures and on complete rankings modeled as
  permutations. The presented results include a linear-time algorithm
  for finding a minimum cost decomposition for simple cycles and a
  linear-time $4/3$-approximation algorithm for permutations that
  contain multiple cycles.  The proposed methods rely on investigating
  a newly introduced balancing property of cycles embedded in trees,
  cycle-merging methods, and shortest path optimization techniques.
  \keywords{Permutations, rankings, similarity distance, transposition
    distance}
\end{abstract}

\section{Introduction}
\label{intro}

Meta-search engines, recommender platforms, social data aggregation
centers as well as many other data processing systems are centered
around the task of ranking distinguishable objects according to some
predefined criteria~\cite{Ailon08,Renda03,Ricci01}. Rankings are
frequently provided by different experts or search engines, and
generated according to different ordering approaches. To perform
comparative studies of such rankings or to aggregate them, one needs
to be able to assess how much they agree or disagree. This is most
easily accomplished by assuming that data is given in the form of
complete rankings -- i.e., permutations -- and that one ranking may be
chosen as a reference sample (identity). In this case, the problem of
evaluating agreements between permutations essentially reduces to the
problem of sorting permutations.

The problem of sorting distinct elements according to a given set of
criteria has a long history and has been studied in mathematics,
computer science, and social choice theory
alike~\cite{Goulden04,Hofri95,Lint01}. One volume of the classical
text in computer science -- Knuth's \emph{The Art of Computer
  Programming} -- is almost entirely devoted to the study of
sorting. The solution to the problem is well known when the sorting
steps are swaps (transpositions) of two elements: In this case, it is
convenient to first perform a cycle decomposition of the permutation
and then swap elements in the same cycle until all cycles have unit
length.

Sorting problems naturally introduce the need for studying \emph{distances} between permutations. There are many different forms of distance functions on permutations, with the two most frequently used being the Cayley distance and the Kendall distance~\cite{diaconis1977spearman}. 
%In the former case, each swap of pairs of elements contributes equally to the overall distance, while in the latter case, each swap of two adjacently placed elements has the same bearing on the value of the distance. 
Although many generalizations of the Cayley, Kendall and other
distances are known~\cite{kumar2010generalized}, only a handful of
results pertain to distances in which one assigns positive weights or
random costs\footnote{Throughout the paper, we use the words cost and
  weight interchangeably, depending on the context of the exposition.}
to the basic rearrangement
steps~\cite{angelov2008sorting,Kapah,Farnoud12}. Most such work has
been performed in connection with genome rearrangement and fragile DNA
breakage studies~\cite{Bafna98,fertin2009combinatorics} and for the
purpose of gene prioritization~\cite{tranchevent2008endeavour}. Some
other examples appear in the social sciences literature~(see
references in \cite{Farnoud12}), pertaining to constrained vote
aggregation and logistics~\cite{Huth09}.

\begin{figure}
\begin{center}
\includegraphics[width=4.3in]{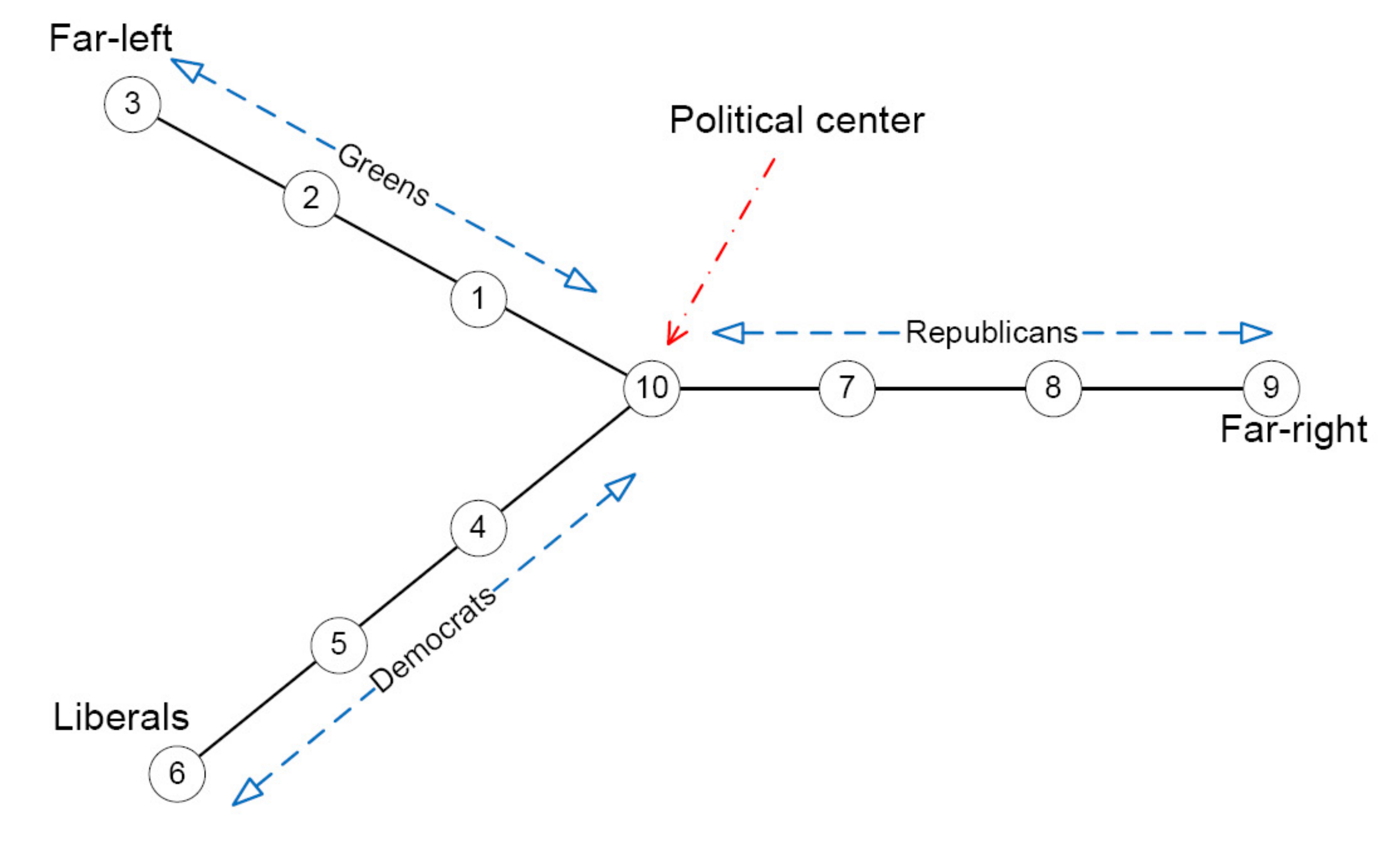}
\caption{A ``Y-tree'' governing the cost of transposing political candidates: Swapping candidates closer in their political beliefs introduces a smaller change in the overall distance between rankings as compared to swapping candidates from different parties.}
\label{fig:parties}
\end{center}
\end{figure}

A number of practical problems call for positive costs (weights) on
transpositions, and costs that capture some constraint on the
structure of the transpositions. The problem at hand may then be
described as follows: For a given set of positive costs assigned to
transpositions of distinct elements, find a smallest cost sequence of
transpositions converting a given permutation to the identity.

In our subsequent analysis, we focus on constraints that take into account that candidates may be similar and that transposing similar candidates should induce a smaller cost than transposing dissimilar candidates. We refer to the underlying family of distance measures as \emph{similarity distances}. The similarity distance is not to be confused with the distance used in~\cite{sculley2007rank}, where the goal was to \emph{rank similar items} close to each other in an aggregated list.

To illustrate the practical utility of the notion of similarity, we present next a number of illustrative examples.

The first example comes from social choice theory. When ranking politicians and assessing the opinion dynamics of voters, one often needs to take into account that the candidates come from different parties. Swapping candidates from the same party may be perceived as having a smaller impact on the overall diversity of the ranking or outcome of an election than doing otherwise. As an example, consider the following three rankings of politicians,
with party affiliations indicated by the letters $D$ and $R$:
\[
\begin{split}
\pi_1=&(\text{Clinton, Obama, Bush, Kerry, Romney}) \rightarrow (D,D,R,D,R),\\
\pi_2=&(\text{Obama, Clinton, Bush, Kerry, Romney}) \rightarrow (D,D,R,D,R), \\
\pi_3=&(\text{Clinton, Bush, Obama, Kerry, Romney}) \rightarrow (D,R,D,D,R).
\end{split}
\]
Notice that $\pi_2$ and $\pi_3$ differ from $\pi_1$ only in one transposition. In the first case, the swap involves two members of the same party, while in the second case, the transposed candidates belong to two different parties.
It would hence be reasonable to assume that the distance between $\pi_1$ and $\pi_2$ is smaller than the distance between $\pi_1$ and $\pi_3$, because in the latter case, the permutations have different overall ordering of the parties\footnote{Clearly, one could also argue that changes at the top of the list are more relevant than changes at the bottom, in which case the comment about the pairwise distances should be reversed. An overview of \emph{positional distances} may be found in~\cite{Farnoud12}, and the related work on generalized \emph{Borda counts} in~\cite{marchant1996valued,young1988condorcet} and references therein.}.

To capture this similarity, candidates may be arranged into a
tree-structure with each edge having a certain weight, so that the
transposition cost of two candidates equals the weight of the unique
path between them. An illustrative example involving three parties is
shown in Fig.~\ref{fig:parties}, where the tree has only one vertex of
degree greater than $2$, corresponding to the political
center. Republicans, Democrats and Greens are all arranged on
different branches of the tree, and in order of their proximity to the
political center. Note that two Republicans are generally closer in
the tree compared to a Republican and a Democratic candidate, implying
that transpositions involving Republicans are, on average, less costly
than those involving candidates of two different parties.

Another application of metric-tree weight distances is in assignment aggregation and rank aggregation~\cite{Ailon08,dwork2001rank,Fagin03,touri2012nonuniform,Liu07,Pihur09}. In the former case, a committee of $m$ members is tasked with distributing $n$ jobs to $n$ candidates. Each committee member provides her suggestion for a full assignment of candidates to jobs. The goal is to aggregate the assignments given by individual committee members into one assignment. If a measure of similarity between the candidates is available, one can use the similarity to aggregate the assignments by finding the best compromise in terms of swapping candidates of similar qualifications, age, gender, working hour preferences, etc. This is achieved by computing the median of the rankings under a suitable similarity condition, such as the metric-path cost~\cite{Farnoud12}. Note that even in this case, the candidates may be arranged into a star-like tree structure reminiscent of Fig.~\ref{fig:parties}.
%, by devoting branches to male candidates arranged by age, female candidates arranged by age or highest degree earned etc.

The third application, and the one that has received the most
attention in the areas of computer science and search engines, is
related to overcoming biases of search
engines~\cite{Ailon08,dwork2001rank,ZW}. As an example, when trying to
identify the links most closely associated with a query, many
different search engines can be utilized, including Google, Yahoo!,
Ask, Bing, IBM, etc. One may argue that the most objective, and hence
least biased, rankings are produced by aggregating the rankings of the
different search engines. Many search queries are performed with the
goal of identifying as many \emph{diverse possibilities} on the first
page or the first two listed pages. Such problem also motivate the
goal of identifying similarities on trees, as many search items may be
naturally arranged in a tree structure. Simulation results proving
that the use of page similarities may lead to more diverse solutions
can be found in our companion paper~\cite{Farnoud12}.

Similarity distances may also be used as valuable tools in gene
prioritization studies. Gene prioritization is a method for
identifying disease-related genes based on diverse information
provided by linkage studies, sequence structure, gene ontology and
other procedures~\cite{endeavor}. Since testing candidate genes is
experimentally costly, one is often required to prioritize the list by
arranging the genes in descending order of likelihood for being
involved in the disease. Different prioritization methods produce
different lists, and similarity of the lists carries information about
which genes may be of importance under different selection
criteria. In addition, since genes are usually clustered into family
trees according to some notion of similarity, finding lists that
prioritize genes while at the same time ensuring that all families of
genes are tested is of great importance.

To conclude this brief motivation, we also point out that similar
star-like trees arise in phylogeny, where each branch contains species
that evolved from each other, and where similarity is captured at the
genomic sequence level~\cite{fertin2009combinatorics}. A higher-level
tree of life itself may be represented by a star-like tree by ignoring
``side-branches'' of certain organisms.

The contributions of this work are three-fold. First, we introduce a
Y-tree (i.e., a tree with at most one node of degree three) cost
function and a notion of similarity between permutations associated
with this special tree structure. In this setting, the cost of
transposing two elements equals the weight of the shortest path in a
Y-tree. Our focus on Y-trees is largely motivated by the fact that the
general tree analysis appears to be prohibitively complex; at the same
time, Y-trees represent computationally manageable and sufficiently
accurate approximations for many tree similarity models. Second, we
describe an exact linear time decomposition algorithm for cycle
permutations with Y-tree costs.  Third, we develop a linear time,
4/3-approximation method for computing the similarity distance between
arbitrary permutations.
%In addition, for permutations whose functional digraphs consist of cycles that may be embedded in the Y-tree in a planar fashion, we describe an exact algorithm for computing the similarity. 
%Note that in the above context, the term ``embedding'' differs from the standard notion of graph embedding and is related to the work in~\cite{Bienstock90}.

The paper is organized as follows. Section \ref{sec:pre} introduces the notation and definitions used throughout the paper. Section \ref{sec:path} contains a brief review of prior work as well as some relevant results used in subsequent derivations. This section also presents a linear time algorithm for computing the Y-tree similarity between cycle permutations. This algorithm is extended in Section \ref{sec:GP} to general permutations via cycle-merging strategies that provide linear time, constant-approximation guarantees. Section 5 contains the concluding remarks.

\section{Mathematical Preliminaries}
\label{sec:pre}
For a given ground set $[n]\triangleq \{1, 2, \ldots, n\}$, a permutation $\pi:[n]\to[n]$ is a bijection on and onto $[n]$. The collection of all permutations on $[n]$ -- the symmetric group of order $n!$ -- is denoted by $S_n$. 

There are several ways to represent a permutation. The two-line
representation has the domain written out in the first line and the
corresponding image in the second line. For example, the following
permutation is given in two-line form:
\[
\pi =
\begin{pmatrix}
1 & 2 & 3 & 4 & 5 & 6\\
6 & 1 & 2 & 5 & 4 & 3
\end{pmatrix}.
\]
The one-line representation is more succinct as it only utilizes the second row of the two-line representation; the above permutation in one-line format reads as $(6, 1, 2, 5, 4, 3)$. The symbol $e$ is reserved for the identity permutation $(1, 2,\ldots, n)$.

Sometimes, we find it useful to describe a permutation in terms of
elements and their images: hence, a third description of the
aforementioned permutation is $\pi(1)=6,\, \pi(2)=1,\, \pi(3)=2,\,
\pi(4)=5,\, \pi(5)=4,$ and $\pi(6)=3$. A straightforward
interpretation of these expressions is that $\pi(i)$ represents the
element placed in position $i$. We also define the inverse of a
permutation $\pi$, $\pi^{-1}$, in which $\pi^{-1}(i)$ describes the
position of element $i$. With this notation at hand, the product of
two permutations $\pi, \sigma\in S_n$, $\mu=\pi \, \sigma$, can be
defined by $\mu(i)=\pi(\sigma(i))$, for all $i\in [n]$. The
\emph{support} of a permutation $\pi \in S_n$, written $\supp(\pi)$, is the
set of all $i \in [n]$ with $\pi(i) \neq i$. We write $|\pi|$ to refer to
$|\supp(\pi)|$.

Permutations may be used in a natural way to describe rankings over any set of distinct elements $P$ of cardinality $n$ by imposing an ordering on the set of elements. As an illustration, for the set $P=\{$Clinton, Bush, Obama, Kerry, Romney$\}$, one may order the names lexicographically as 
\[(\text{Bush, Clinton, Kerry, Obama, Romney})\] 
and subsequently assign numerical values to the elements according to this ordering as Bush=1, Clinton=2, etc. Hence, the ranking $(\text{Kerry, Obama, Clinton, Romney, Bush})$ corresponds to a permutation that reads as $(3,4,2,5,1)$.

For $k>1$, a \emph{k-cycle}, denoted by $\kappa=(i_{1}\,\ldots\, i_{k})$, is a permutation that acts on $[n]$ in the following way\footnote{This is not to be confused with the one line representation using commas between entries.}:
\[
i_{1} \to i_{2} \to  \ldots  \to i_{k}  \to i_{1},
\]
where $x\to y$ denotes $y=\kappa(x)$. In other words,
$\kappa=(i_{1}\,\ldots\, i_{k})$ cyclically shifts elements in the
permutation confined to the set $\{i_1,\ldots,i_k\}$ and keeps all
other elements fixed.  A cycle of length $2$ is called a
\textit{transposition}, and is denoted by $(a\,b)$.
%An \textit{adjacent transposition} is the transposition of two adjacent elements. 

In general, for $a, b\in [n]$, $\pi(a\,b)\not=(a\,b)\pi$, because
$\pi(a\,b)$ corresponds to swapping elements of $\pi$ in positions $a$
and $b$ while $(a\,b)\pi$ corresponds to swapping elements $a$ and $b$
in $\pi$. For instance, $(6, 1, 2, 5, 4, 3)(2\,3)=(6, 2, 1, 5, 4, 3)$,
while $(2\,3)(6, 1, 2, 5, 4, 3)=(6, 1, 3, 5, 4, 2)$. Note that in the
former example, we used $\pi(a\,b)$ to denote the product of a
permutation and a transposition.
%, which is not to be confused with the image of an element under $\pi$.

Two cycles are said to be \textit{disjoint} if the intersection of their supports is empty; furthermore, two cycles are termed to be \textit{adjacent} 
if they have exactly one common element in their supports. Although non-disjoint cycles are sporadically mentioned in the combinatorial literature, their use is extremely limited due to the fact that disjoint cycles offer simpler means to study problems on permutations. In particular, the concept of adjacent cycles was, to the best of the authors' knowledge, not previously used for analyzing sorting algorithms.
 
A permutation can be uniquely decomposed into a product of disjoint cycles, often referred to as the \textit{cycle decomposition} or the \textit{cycle representation}. For example, the cycle decomposition of the permutation $(6,\,1,\,2,\,5,\,4,\,3)$ 
equals $(1\,6\,3\,2)(4\,5)$, where one can freely choose the order in which to multiply $(1\,6\,3\,2)$ and $(4\,5)$. We note that a 
cycle may also be written as a product of shorter cycles comprising a combination of disjoint and adjacent cycles. We term the result of this procedure an \textit{adjacent cycle decomposition}. One significant difference between the two aforementioned cycle decompositions is that in an adjacent cycle decompositions, the order of multiplication matters (i.e., the product is non-commutative); $(1\,6\,3\,2)$ equals $(2\,1\,6)(3\,6)$, but not $(3\,6)(2\,1\,6)$. As opposed to the disjoint cycle decomposition which is unique, there may exist multiple adjacent cycle decompositions of a given permutation.

The \textit{functional digraph} of a function $f: [n]\to[n]$, denoted by $\mathcal{G}(f)$, is a directed graph with vertex set $[n]$ and arcs from $i$ to $f(i)$ for each $i\in [n]$. Arcs are subsequently denoted by $(i \to f(i))$. For a permutation $\pi$, $\mathcal{G}(\pi)$ is a collection of disjoint cycles; hence, the cycles of the permutation correspond to the cycles of its functional digraph.

Given any connected, undirected, edge-weighted graph $G$ on the vertex
set $[n]$ with positive edge weights, we can define a metric $\varphi$
by letting $\phi(a,b)$ be the minimum weight of an $a,b$-path in
$G$. If $\varphi$ can be defined from $G$ in this way, we say that
$\varphi$ is a \emph{graph metric} and that $G$ is a \emph{defining
  graph} for $\varphi$. Any metric on a finite set is a graph metric:
one can let $G$ be a complete graph where the weight of each edge
$(a,b)$ equals $\varphi(a,b)$. However, $\phi$ may have other, sparser
defining graphs. We will typically be interested in graph metrics with
a defining graph that falls into some special graph class. When $\phi$
is a metric on $[n]$, we also consider $\phi$ as giving weights to the
transpositions in $S_n$, where the transposition $(a\,b)$ has weight
$\phi(a,b)$.

%An arbitrary chosen minimum weight path between two vertices $a$ and $b$ will subsequently be denoted by $p_{\varphi}^*(a, b)$.
% An example is given in Fig.~\ref{fig:dg}.
%\begin{figure}
%\centering
%\includegraphics[width=5cm]{DG.pdf}
%\caption{The defining graph for $\varphi$, where the weights of \emph{edges} take the values $\text{wt}(1,2)=\text{wt}(3,5)=\text{wt}(4,5)=1, \text{wt}(2,5)=\text{wt}(3,4)=2, \text{wt}(2,3)=3$, and $\infty$ for all other pair of vertices. The cost of transposing $i$ and $j$, $\varphi{(i\,j)}$, equals the minimum weight path between $i$ and $j$ in the defining graph.}\label{fig:dg}
%\end{figure}

The weight of an ordered sequence of transpositions is defined as the sum of the weights of its constituent elements. That is, the weight of the sequence of transpositions $T=(\tau_{1}, \ldots,\tau_{|T|})$ equals
\[
\wtphi(T)=\sum_{i=1}^{|T|} \varphi (\tau_{i})=\sum_{i=1}^{|T|} \varphi (a_i,b_i),
\]
where we used $\tau_i$ to denote the transposition $(a_i,b_i)$, and
$|T|$ to denote the number of transpositions in the sequence $T$.
When $\phi$ is understood (as will typically be the case throughout
this paper) we suppress the subscripts and simply write $\wt(T)$.  The
same convention is used for all other notation involving the subscript
$\phi$.

If $\sigma=\pi\tau_{1}\tau_{2}\ldots\tau_{|T|}$, we refer to
\(T=(\tau_{1}, \ldots,\tau_{|T|})\) as a \emph{transform}, converting
\(\pi\) into \(\sigma\). The set of all such transforms is denoted by
$A(\pi,\sigma)$. Clearly, $A(\pi,\sigma)$ is non-empty for any $\pi,
\sigma\in S_n$. A transform that converts $\pi$ into $e$, the identity
permutation, is a \emph{sorting} of $\pi$. On the other hand, a
\emph{decomposition} of $\pi$ is a sequence $T=(\tau_{1},
\ldots,\tau_{|T|})$ of transpositions such that
$\pi=\tau_{1}\tau_{2}\ldots\tau_{|T|}$. Note that the minimum weight
of a decomposition is the same as the minimum weight of a sorting as
one sequence is equal to the other in reverse order.
%Permutation distances can be naturally defined in terms of uniform transposition weights, as is the case with the Kendall and Cayley distance. Unlike the aforementioned notions of distance, where every allowed transposition has identical weight, we hereafter consider the case where weights associated with transpositions are nonuniform across transpositions. For instance, weighted versions of the Cayley distance and Kendall-$\tau$ distance were described in our companion paper~\cite{Farnoud12} for the purpose of capturing positional relevance of ranked candidates. The Y-tree distance measure is more versatile when it comes to

The $\varphi$-weighted transposition distance between $\pi$ and $\sigma$ is defined by
\[
\distphi (\pi, \sigma)=\min_{T\in A(\pi,\sigma)} \wtphi(T).
\]
Computing $\dist (\pi, \sigma)$ may be cast as a minimization problem over $A(\pi,\sigma)$, namely the problem of finding a \textit{minimum cost transform} $T^*\in A(\pi,\sigma)$ such that $\dist (\pi, \sigma)=\wt(T^*)$. 
If $\varphi(a,b)=1$ for all distinct $a$ and $b$, the weighted 
transposition distance reduces to the well-known Cayley distance.

It is easy to verify that for every positive weight function, the weighted transposition distance $\dist $ is a metric
%\footnote{A pseudo-metric is a generalized metric in which two distinct points may be at distance zero.} 
and furthermore, left-invariant (i.e., $\dist (\pi,\sigma)=\dist (\omega \, \pi,\omega \, \sigma)$). Hence, we may set one
of the permutations (say, $\sigma$) to $e$, and write  
\[
\distephi(\pi)=\distphi (\pi, e)=\min_{T\in A(\pi,e)} \wtphi(T).
\]
We refer to the problem of computing $\diste (\pi)$ as the (weighted) \emph{decomposition problem}.
%Although at this point it is not known whether the problem is NP hard or not, at first glance, it appears to be computationally difficult , due to the fact that it is related to finding minimum generators of groups and the subset-sum problem\cite.

With respect to the choice of weight functions, we restrict our attention to the previously introduced family of graph metric weights, satisfying the triangle inequality
\[
\varphi (a,b) \le \varphi (a,c)+\varphi (c,b), \text{\,for\, all distinct\;} a, b, c\in [n].
\]
In particular, if we fix a tree-structured defining graph, the weight function $\varphi$ is termed a \emph{metric-tree} weight function. For such defining graphs, there clearly exists a unique minimum cost path between any two vertices, and for $a, b\in [n]$, $\varphi(a,b)$ is the sum of the weights of the edges on the unique path between $a$ and $b$ in $G$. If $G$ is a path (line graph), then $\varphi$ is called a \textit{metric-path} weight function. If there exists a unique vertex in a tree-structured $G$ of degree larger than or equal to three, the graph is called a \textit{metric-star}. The vertex with highest degree is referred to as the \textit{central vertex}. If the central vertex has degree three, the defining graph is called a \textit{Y-tree}. The corresponding metric is referred to as the Y-tree metric.
%\footnote{The weighted tree $\mathcal{G}$ is different from the defining graph of the weight function, which is a complete graph. Actually, the latter can be treated as a special case of the defining graph with the understanding that though they may exist multiple paths between two vertices, every path has the same weights. So it is enough to represent such weight function by a tree structure. In our analysis to follow, we refer such weighted tree structure as defining graph.}
Examples of the aforementioned defining graphs are shown in Fig.~\ref{fig:mt}. 

\begin{figure}
\centering
\mbox{
\subfigure[A defining graph $G$ corresponding to a general metric-tree weight function $\varphi$. The edge labels correspond to their underlying weights. From the graph, one reads $\varphi{(1,7)}=\varphi{(1,3)}+\varphi{(3,4)}+\varphi{(4,7)}=2+3+4=9$.]{\includegraphics[width=6.5cm]{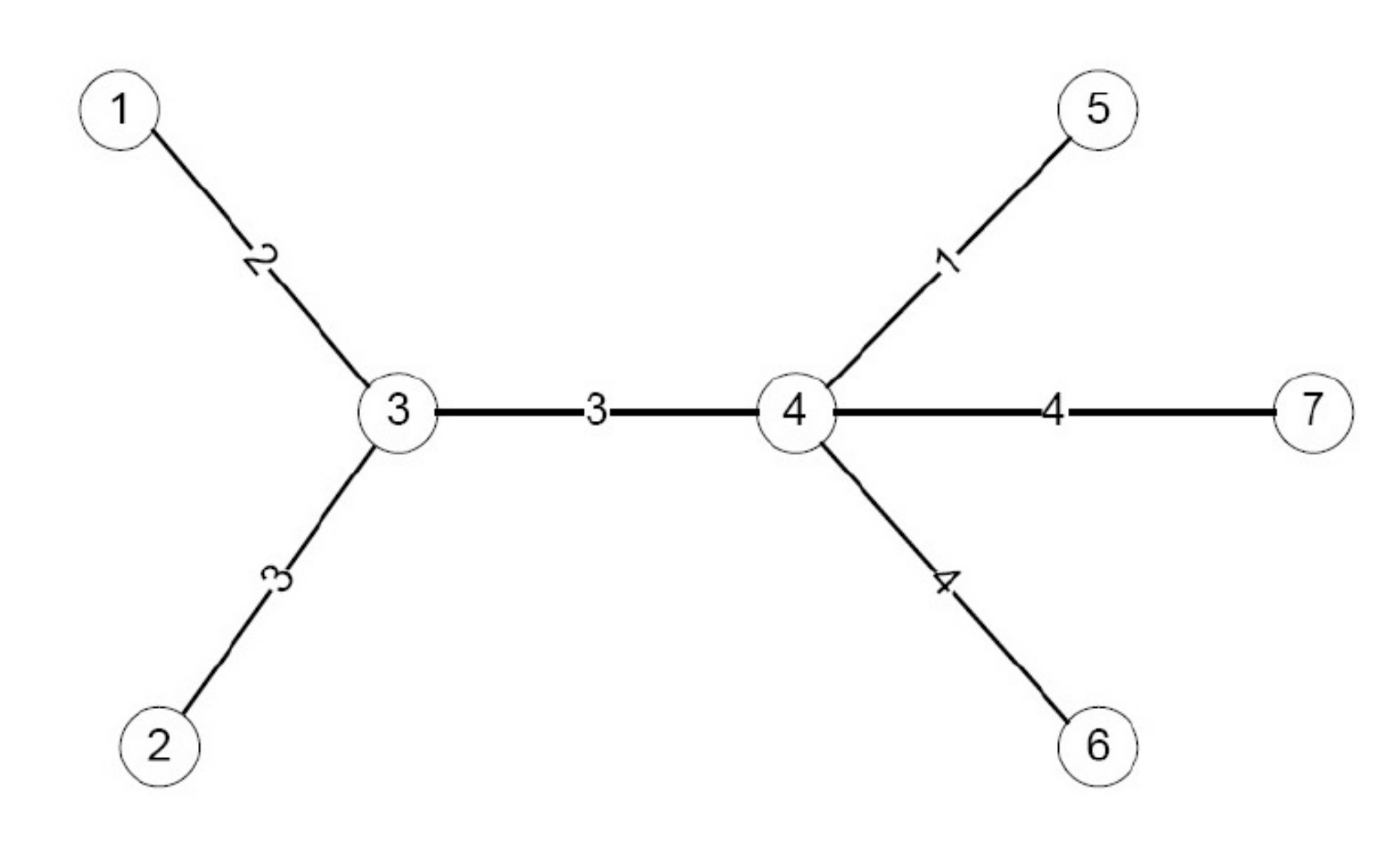}}\quad
}
\mbox{
\subfigure[A defining graph $G$ of a Y-tree with one edge cost equal to two, and all other edge costs equal to one. From the graph, one reads $\varphi{(1,7)}=\varphi{(1,8)}+\varphi{(8,6)}+\varphi{(6,7)}=1+1+1=3$.]
{\includegraphics[width=6.5cm]{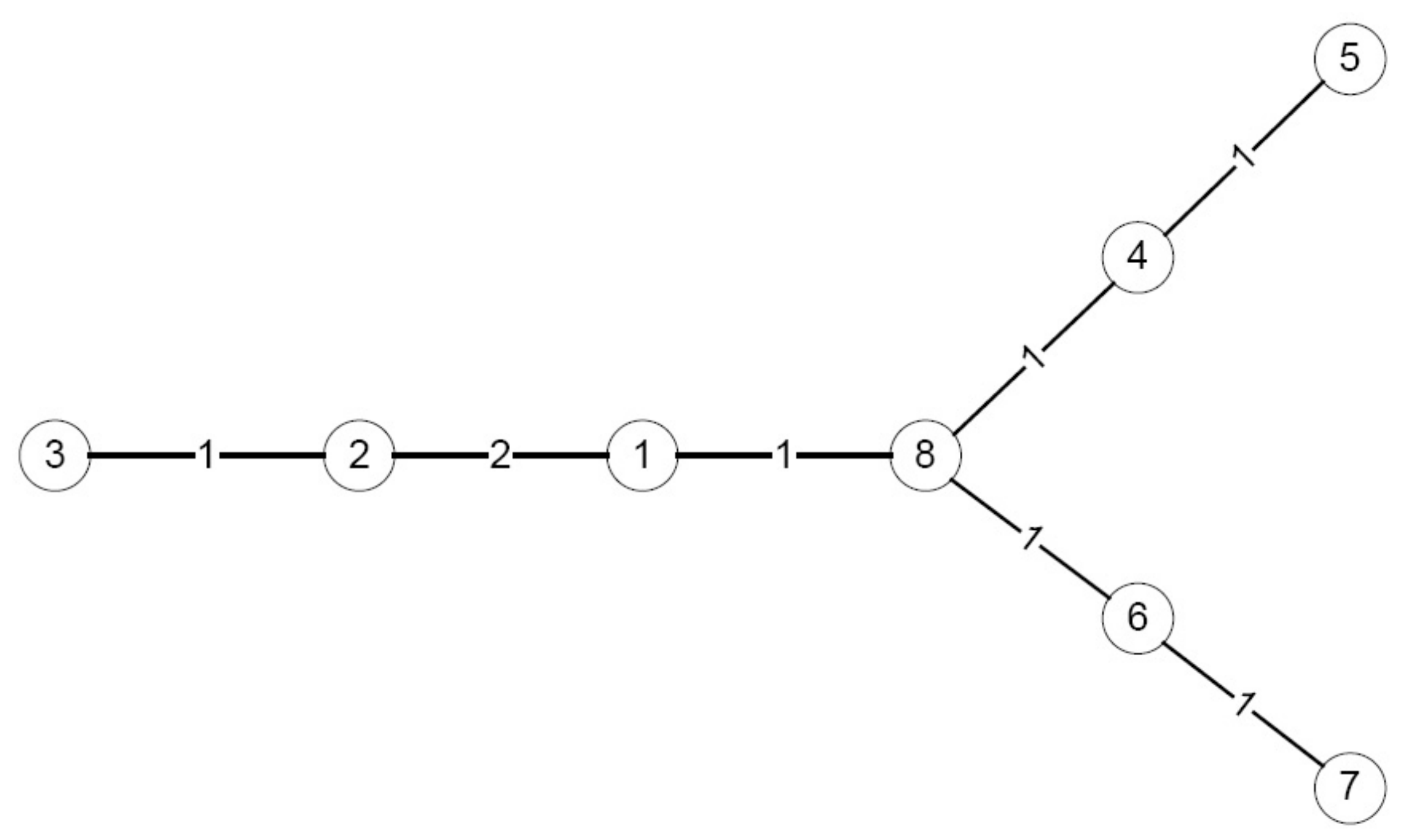}}\quad
}
\caption{Examples of defining graphs.}
\label{fig:mt}
\end{figure}

The problem of finding $\diste (\pi)$ when $\varphi$ is a metric-path
weight was studied by the authors in~\cite{Farnoud12}. The focus of
the results to follow is on determining $\diste (\pi)$ when the
defining graph is a Y-tree. The problem of evaluating $\diste (\pi)$
under a general metric-tree model appears difficult to handle by
methods proposed in this work and will hence not be discussed.

The following function, termed the \emph{displacement}, is of crucial importance in our analysis of similarity distances on Y-trees:
\[
 \Dispphi (\pi,\sigma)=\sum_{i=1}^n \, \varphi (\pi^{-1}(i),\sigma^{-1}(i)).
 \]
The displacement $ \Disp (\pi,\sigma)$ captures the overall cost of independently performing optimal transpositions of pairs of elements that are out of order.
It is again easy to verify that for every positive weight function, the displacement $\Disp (\pi,\sigma)$ is a metric and in addition, left-invariant (i.e., $\Disp (\pi, \sigma)=\Disp (\omega \, \pi, \omega \, \sigma)$, for all $\pi,\sigma,\omega \in S_n$). 
As a result, the notation and analysis may be simplified by assuming that $\sigma=e$ and by denoting the resulting displacement by $\Disp (\pi)$. 

The following properties of the displacement are easy to verify:
\begin{enumerate}
\item $\Disp (\pi)=0$ if and only if $\pi=e$.
\item $\Disp (\pi_1\, \pi_2) \leq \Disp (\pi_1)+\Disp (\pi_2),$ for all permutations $\pi_1$ and $\pi_2$.
\item $\Disp (\pi)=\Disp (\pi^{-1}),$ for all permutations $\pi$.
\end{enumerate}
Consequently, we may write
\[
 \Disp (\pi)=\sum_{i=1}^n \, \varphi (i,\pi(i)).
 \]

The main results of the paper are devoted to the study of decompositions of single cycles, as more general permutation decompositions may be obtained via individual cycle decompositions.

For ease of exposition, we draw the digraph of a permutation and the undirected defining Y-tree graph of the given weight function on the same vertex set, as shown in Fig.~\ref{fig:Yfunctional}. In this case, we say that the permutation is embedded in the defining graph. This graphical representation of both the cost function and the cycle decomposition of a permutation allows us to illustrate examples and gain intuition about the algorithms involved in the decomposition approach.
\begin{figure}
\begin{center}
\includegraphics[width=6.5cm]{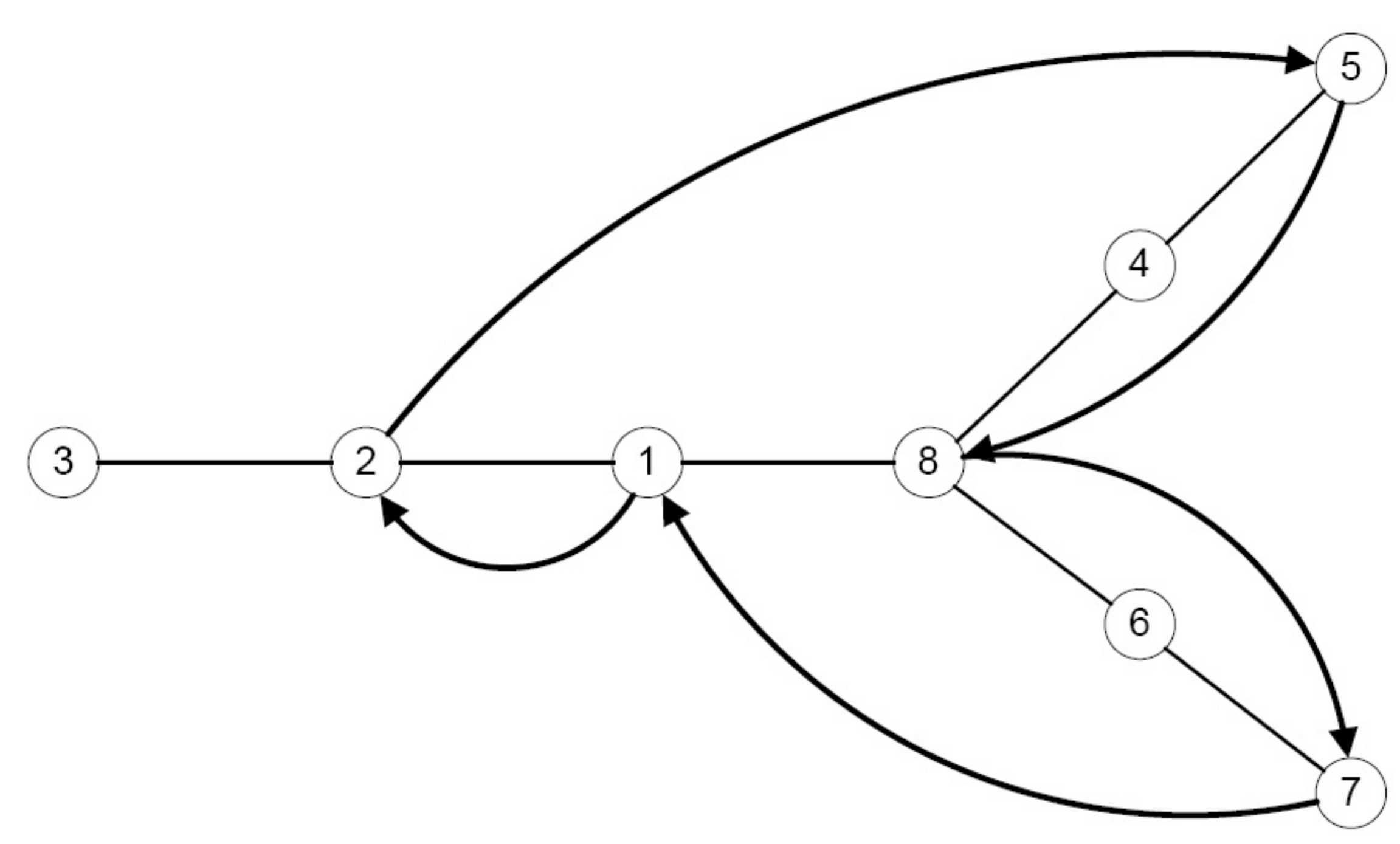}
\caption{Defining Y-tree and the cycle $(1\,2\,5\,8\,7)$. Thin lines represent the defining Y-tree $G$, while boldfaced arcs represent the digraph of the cycle permutation, $\mathcal{G}$.}
\label{fig:Yfunctional}
\end{center}
\end{figure}

Denote the branches of a Y-tree, which are sets of nodes on paths
starting from the central vertex and extending to a leaf, excluding
the central vertex, by $B_1$, $B_2$, and $B_3$. Furthermore, for ease
of exposition, denote the branch containing vertex $v$ by
$\br(v)$. Next, we formalize the notion of a cycle lying on a path on
the Y-tree as a cycle that has support contained in $B_i \bigcup B_j
\bigcup \{{\cv\}}$, for some not necessarily distinct $i,j$, and with
$\cv$ representing the central vertex. In other words, a cycle lies on
a path if its support is contained in the union of at most two of the
three branches and the central vertex. For a branch pair \((B_{i},
B_{j})\), $i \neq j$, let \(l^\pi_{ij}\) be the number of arcs from
\(B_{i}\) to \(B_{j}\) in $\pi$; similarly, let \(l^\pi_{ji}\) be the
number of arcs from \(B_{j}\) to \(B_{i}\) in $\pi$. If it is clear
from the context, the superscript $\pi$ will be omitted.

For a cycle permutation, we say that the branch pair $(B_i,B_j)$ is balanced if \(l_{ij}=l_{ji}\). If \(l_{ij}=l_{ji}\) for all $i,j \in \{1,2,3\}$, we say that the \textit{cycle is balanced}.

The \textit{inefficiency} of a transposition $\tau=(a\,b)$ with
respect to a permutation $\pi$ and for a given cost function
$\varphi$, denoted by $\ineff \varphi a b \pi$ and by $\ineff* \varphi \tau \pi$, is defined as
\[
\ineff \varphi a b \pi=2\varphi(a,b)-(\Disp (\pi)-\Disp (\pi(a\,b))).
\]

The intuition behind the notion of inefficiency comes from the observation that a transposition $(a \,b)$ can reduce the overall displacement by at most $2\varphi(a,b)$; the inefficiency measures the gap from the optimal reduction. Also, since $2\varphi(a,b)=\Disp ((a,b))$, it follows that the inefficiency is nonnegative. Henceforth, a transposition $(a\,b)$ is termed \emph{efficient} with respect to $\pi$ if $\ief a b \pi=0$ and \emph{inefficient} if $\ief a b \pi>0$.

The proposed algorithm for finding a minimum cost decomposition of a permutation under Y-tree weights consists of two steps:
\begin{enumerate}
\item First, we derive a closed form expression for the minimum cost of a decomposition of a single cycle and present an exact algorithm that can find the minimum cost decomposition $T^*$ in linear time. 
\item Second, for general permutations with multiple cycles, we develop a linear time, $4/3$-approximation algorithm that uses decompositions of single cycles. 
\end{enumerate}

\section{Similarity Distances on Y-trees: The Single Cycle Case}\label{sec:path}

The gist of the proposed approach for computing the similarity distance on a Y-tree is to decompose a cycle in such a way that all its components are supported on paths. Once such a decomposition is performed, we can invoke the results of our companion paper~\cite{Farnoud12}, which asserts that cycle decompositions for metric-path costs can be performed optimally in linear time. 
The key question is hence to determine if one can perform a decomposition of an arbitrary cycle into cycles that are supported on paths in an efficient manner, i.e., by only using efficient transpositions.
For this purpose, we find the following lemma that applies to \emph{general permutations} useful.

\begin{lm}\label{lem:metric-path}
Let $\varphi$ be a metric-tree weight function, and let \(\pi \) be a permutation. The minimum decomposition cost of $\pi$ is bounded below by one half of its displacement, i.e.,
\[
\diste (\pi)\ge\frac{1}{2} \Disp (\pi).
\]
The lower bound may be achieved for metric-path weight functions $\varphi$, for which
\begin{equation}
\diste (\pi)=\frac {1}{2}\Disp (\pi).\label{path-case}
\end{equation}
\end{lm}

The proof of the previous lemma can be found in our companion paper~\cite{Farnoud12}, with the latter claim following by induction on the number of elements in the support of the permutation $\pi$. 

An algorithm which describes how to find a minimum cost decomposition $T^*$ in this case can be easily devised using the idea behind the proof, and is presented next. 

Without loss of generality, label the vertices in the defining path from left to right as $1, 2, \cdots, n,$ and suppose that we are decomposing a single cycle $\pi=(v_1\, v_2\, \ldots\, v_{|\pi|})$, with $v_1=\min\; \supp(\pi)$. If this is not the case, rewrite $\pi$ by cyclically shifting its elements. Let $v_t=\min_{i\in \supp(\pi)} \{i: i\not=v_1\}$. With this notation at hand, the steps of the decomposition procedure are listed in Algorithm~\ref{alg:ptd}.

\begin{algorithm}
\caption{path-td}
\label{alg:ptd}
\SetKwFunction{concat}{Concatenate}
\SetKwFunction{this}{path-td}
\tcc{Transposition decomposition of cycles for metric-path weights with defining path given by $1,\dotsc,n$}
\KwIn{A cycle $\pi=(v_1\, v_2\, \ldots\, v_{|\pi|})$, with $|\pi|\ge2$ and $v_1=\min \supp(\pi)$}
\KwOut{A minimum cost decomposition $T=(\tau_1,\dotsc, \tau_{|T|}$) of $\pi$, so that $\pi = \tau_1\dotsm\tau_{|T|}$}
\BlankLine
\lIf{$|\pi|=2$}{\Return $(\pi)$}
$v_t\gets\min \left(\supp(\pi)\backslash \{v_1\}\right)$\;
\eIf{$v_t\not= v_{|\pi|}$}
{
  $\pi_1\gets(v_{1}\,v_{2}\,\cdots\, v_{t})$\;
  $\pi_2\gets(v_{t}\,v_{t+1}\,\cdots\, v_{|\pi|})$
  \tcc*[r]{$\pi=\pi_1\pi_2$}
  \Return\concat(\this($\pi_1$), \this($\pi_2$))\;
}
{
  $\pi_1\gets(v_2\,v_3\,\cdots\, v_{|\pi|})$\;
  $\tau\gets(v_1\,v_{|\pi|})$
  \tcc*[r]{$\pi=\pi_1\tau$}
  \Return \concat(\this($\pi_1$),$\tau$)\;
}
\end{algorithm}

At each call of Algorithm \ref{alg:ptd}, the cycle $\pi$ is rewritten as one of two possible cycle products, depending on whether $v_t=v_{|\supp(\pi)|}$ holds or not. Intuitively, the decomposition breaks cycles using vertices closest to each other, which clearly minimizes the total cost of the transpositions involved. An example of such a decomposition is shown in Fig.~\ref{fig:path}.
\begin{figure}
\centering
\mbox{
\subfigure[We have $v_t=2$ and $v_{|\pi|}=3$, so that $v_t\not=v_{|\pi|}$. Hence, the cycle $(1\,4\,2\,6\,5\,3)$ is decomposed as $(1\,4\,2\,6\,5\,3)=\pi_1\,\pi_2=(1\,4\,2)(2\,6\,5\,3)$. In this step of the decomposition, the arc $(3 \to 1)$ is replaced by two arcs, each belonging to one of the cycles. The two resulting cycles are represented with solid and dashed arcs, respectively.]{\includegraphics[width=10cm]{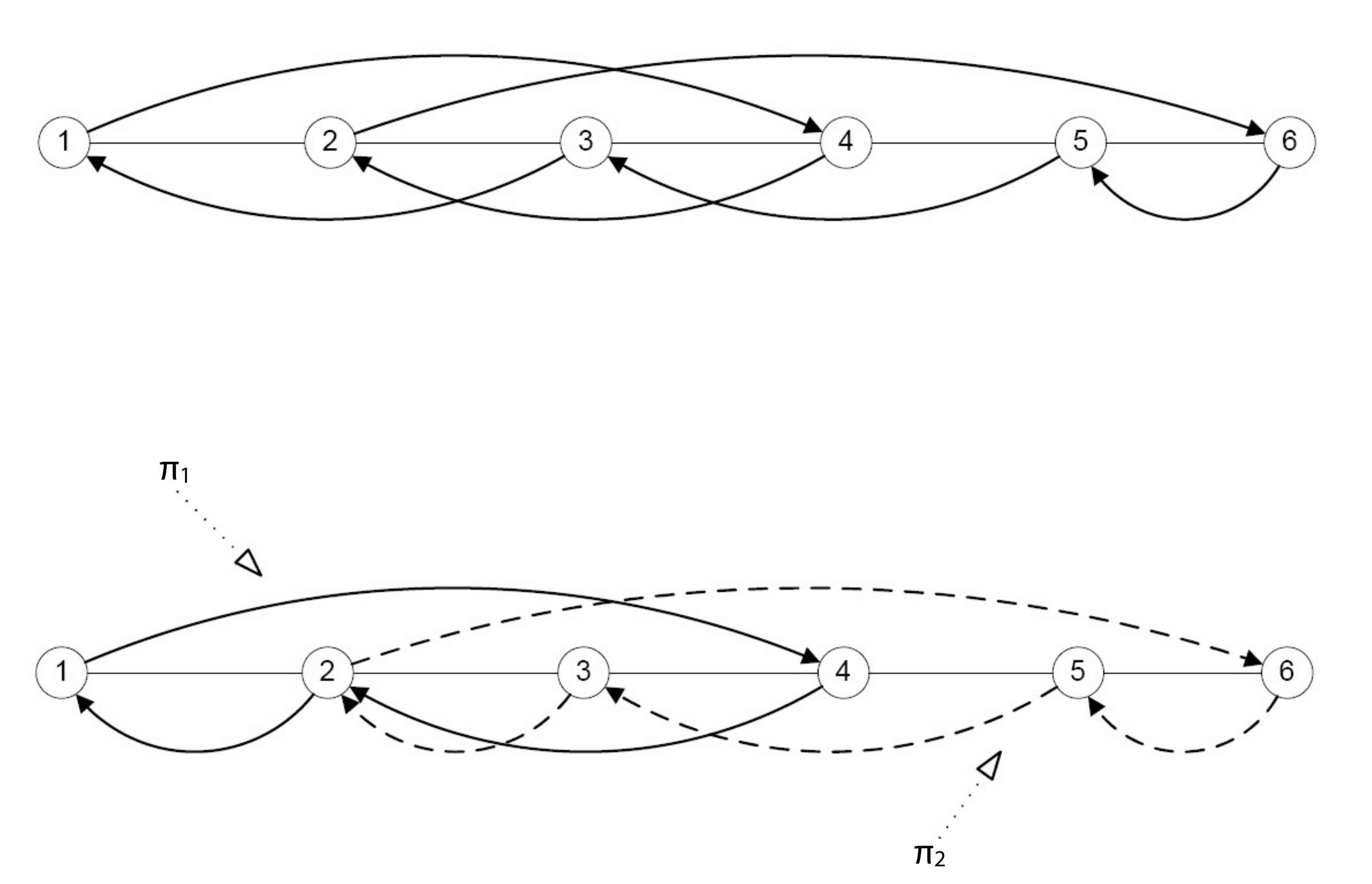}}\quad
}
\mbox{
\subfigure[We have $v_t=v_{|\pi|}=3$. Hence, the cycle $(1\,4\,6\,5\,3)$ is decomposed as $(1\,4\,6\,5\,3)=\pi_1\,\pi_2=(3\,4\,6\,5)(1\,3)$. In this step of the decomposition, the arc $(1 \to 4)$ is replaced by two arcs each belonging to one of the component cycles. The resulting two cycles are represented with solid and dashed arcs, respectively.]{\includegraphics[width=10cm]{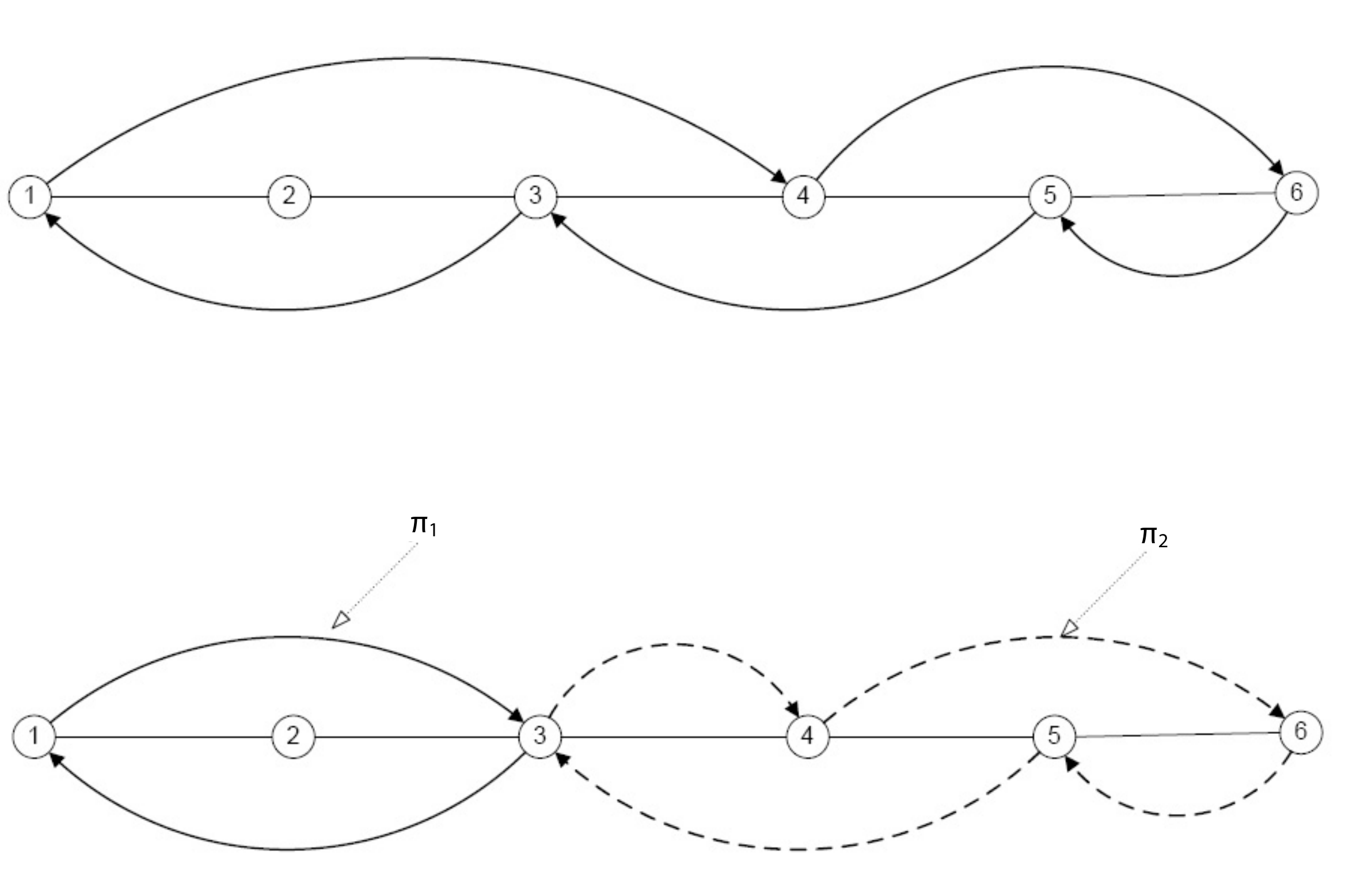}}\quad
}
\caption{Two different decomposition cases encountered in Algorithm~\ref{alg:ptd}, based on whether $v_t=v_{|\pi|}$ holds or not.}
\label{fig:path}
\end{figure}
%\section{Computing the Weighted Transposition Distance for Y Metric-Tree Weight Functions}\label{sec:Y tree}

As illustrated by the cycle in Fig.~\ref{fig:example}, this approach cannot be generalized for Y-tree weight functions. In the example, the total displacement $\Disp ((1\,2\,3))$ equals $6$, while via exhaustive search one can show that
\[\diste ((1\,2\,3))=4\not=\frac{1}{2}\Disp ((1\,2\,3)).\]
\begin{figure}
\begin{center}
\includegraphics [width=5cm]{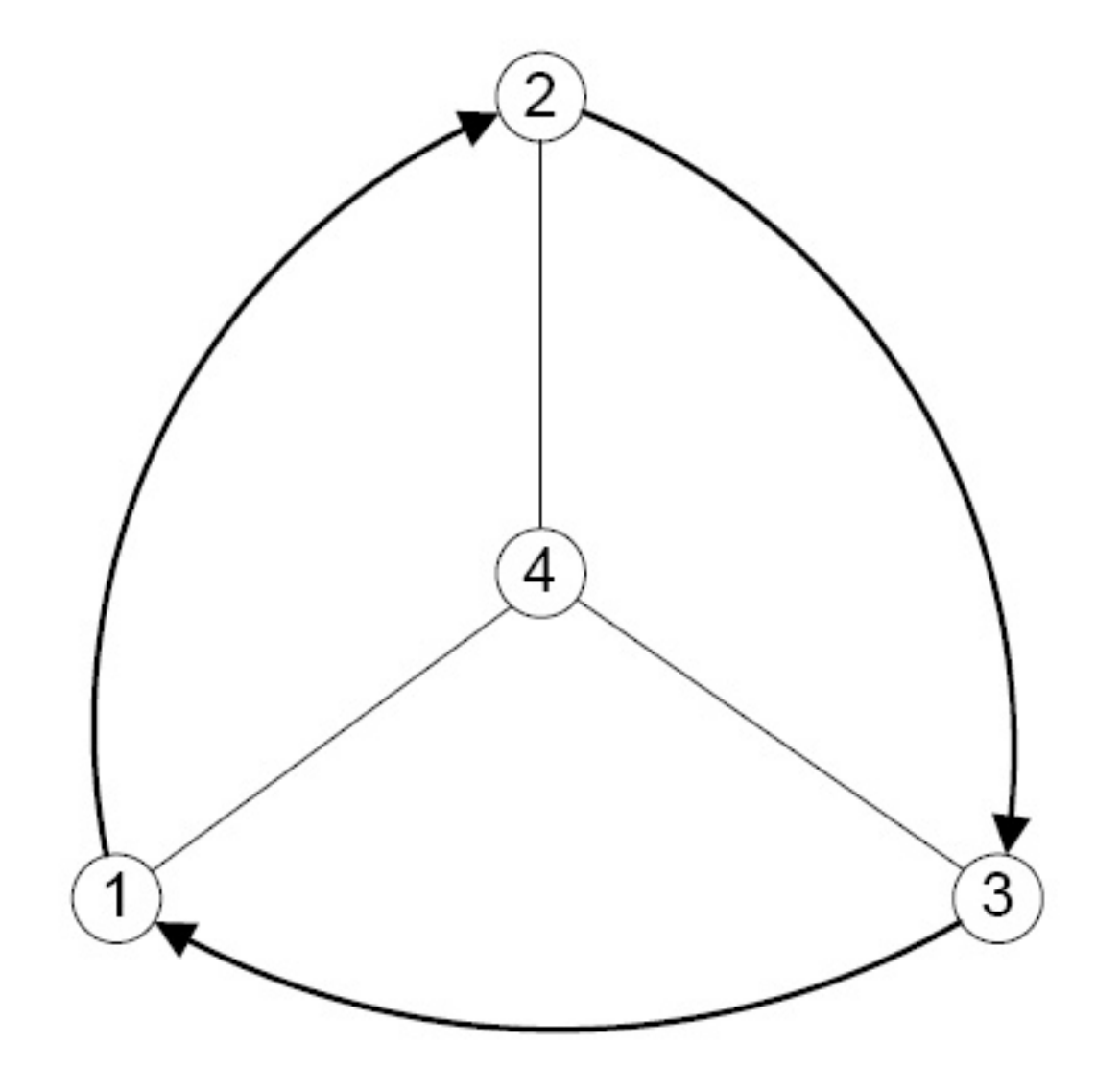}
\caption{The weight $\varphi$ is defined via a Y-tree with all edges of weight one; the cycle equals $(1\,2\,3)$.}
\label{fig:example}
\end{center}
\end{figure}
Note that in Fig.~\ref{fig:example}, the central vertex does not belong to the support of the cycle, and furthermore, the cycle is not balanced. Careful examination hence reveals that in order to generalize Algorithm~\ref{alg:ptd} for Y-tree costs, one has to separately consider three cases: $1\,)$ the case when the central vertex belongs to the support of the cycle; $ 2\,)$ the case when the central vertex does not belong to the support of the cycle, but the cycle is balanced; $3\,)$ the case when neither of the aforementioned two conditions hold.

%The analysis of decomposition algorithms, as already pointed out, relies on using the displacement functions and properties of the cycle embedding in the tree, such as the balancing property. 
We provide next a useful characterization of efficient transpositions. 
%To accomplish this task, we notice that for any cycle embedded in a Y-tree, two vertices $a,b$ may assume only one of the four possible configurations shown in Fig.~\ref{fig:fourtype}. 
To do so, we recall that the defining graph $G$ is a tree, and that
hence there exists a unique path between any two vertices $a,b$ of
$G$. The next lemma describes for which $a$,$b$-paths the
corresponding transposition $(a\,b)$ is efficient.
\begin{lm}\label{lm:efficiency}
Let $G$ be the defining graph of a metric-tree weight function $\varphi$ and let $\pi$ be an arbitrary permutation of length $n$. For distinct $a, b\in [n]$, we have 
\begin{equation}\begin{split}\label{eq:ineff-nonneg}
\ief a b \pi\ge \varphi(a,b)+\varphi(b,\pi(a))-\varphi(a,\pi(a))\ge 0,\\
\ief a b \pi\ge \varphi(a,b)+\varphi(a,\pi(b))-\varphi(b,\pi(b))\ge 0.
\end{split}\end{equation}
 Furthermore, the following claims are equivalent:
\begin{enumerate}[i.]
\item The transposition $(a\,b)$ is efficient.
\item It holds that
\begin{align}
\varphi(a,\pi(a))-\varphi(b,\pi(a))=\varphi(a,b)\label{eq:tri-eq-a},\\
\varphi(b,\pi(b))-\varphi(a,\pi(b))=\varphi(a,b)\label{eq:tri-eq-b}.
\end{align}
\item The vertex $a$ lies on the $(b,\pi(b))$-path in $G$, and the vertex $b$ lies on the $(a,\pi(a))$-path in $G$.
\end{enumerate}
\end{lm}
\begin{proof}
First, note that
\begin{align}\label{eq:ineff}
\ief a b \pi &= 2\varphi(a,b) - \Disp (\pi)+\Disp \left(\pi\left(a\, b\right)\right)\\
&=2\varphi(a,b)-\varphi(a,\pi(a))-\varphi(b,\pi(b))+\varphi(a,\pi(b))+\varphi(b,\pi(a)).\nonumber
\end{align}
From the triangle inequality, one also has
\begin{align}
\varphi(a,b)+\varphi(b,\pi(a))-\varphi(a,\pi(a))\ge0,\label{eq:tri-ieq-a}\\
\varphi(a,b)+\varphi(a,\pi(b))-\varphi(b,\pi(b))\ge0.\label{eq:tri-ieq-b}
\end{align}
By adding \eqref{eq:tri-ieq-a} and \eqref{eq:tri-ieq-b}, and by using \eqref{eq:ineff}, one can show that \eqref{eq:ineff-nonneg} holds as well. Additionally, $\ief a b \pi = 0$ if and only if \eqref{eq:tri-ieq-a} and \eqref{eq:tri-ieq-b} hold with equality, that is, if and only if \eqref{eq:tri-eq-a} and \eqref{eq:tri-eq-b} are true. This proves \eqref{eq:ineff-nonneg}, as well as that claims $i$ and $ii$ are equivalent.

To show that claims $ii$ and $iii$ are equivalent, it suffices to show that $\varphi(a,\pi(a))-\varphi(b,\pi(a))=\varphi(a,b)$ if and only if $b$ is on the path from $a$ to $\pi(a)$. This can be readily verified by inspecting all possible vertex placements as shown in Fig. \ref{fig:b-pos}, and by noting that all weights
on the tree are positive. Note that in Fig. \ref{fig:b-pos}, we have ignored the case where $a$, $b$, and $\pi(a)$ are not all distinct, as this case is particularly simple to check.
\myqed
\end{proof}

\begin{figure}
\begin{centering}
\includegraphics[width=0.75\textwidth]{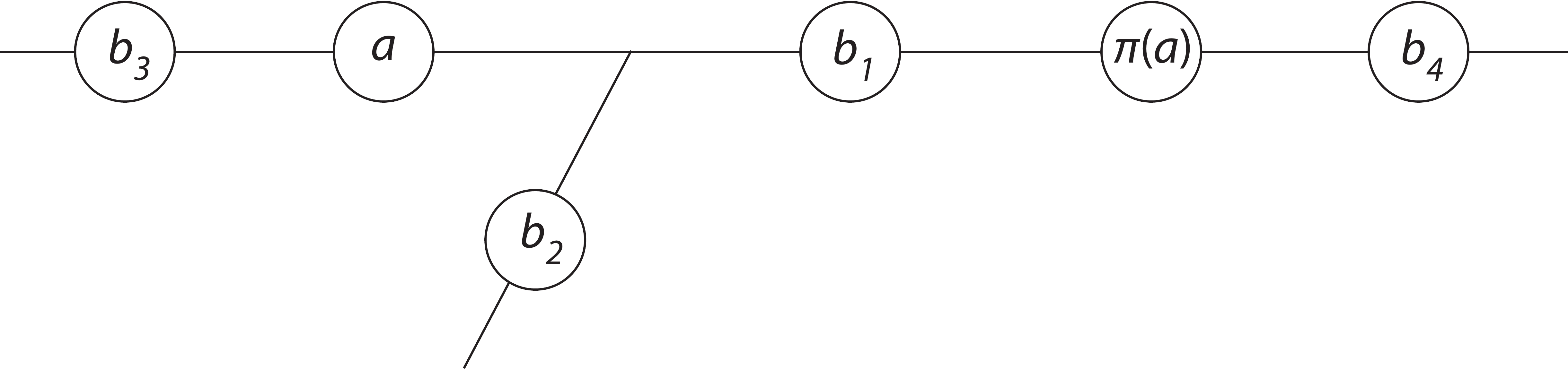}
\par\end{centering}
\protect\caption{There are four possible positions for $b$ with respect to $a$ and
$\pi(a)$ shown by $b_{1}$, $b_{2}$, $b_{3}$, and $b_{4}$.\label{fig:b-pos}}
\end{figure}

The next lemma strengthens  the results of Lemma~\ref{lem:metric-path}, and will be of use in the derivations to follow. 
\begin{lm}\label{lm:lm1-improv}
For a permutation $\pi$, the gap between $\diste (\pi)$ and $\frac12\Disp(\pi)$ equals the sum of the inefficiencies of the transpositions in a minimum weight decomposition.
\end{lm}
\begin{proof}
Let $T^*=(\tau_{1},\dotsc,\tau_{|T^{*}|})$, $\tau_j=(a_j\,b_j)$, be a minimum weight sorting and let $\pi_j=\pi_{j-1}\tau_j$, with $\pi_0 = \pi$. 
For all $j$, we have
\[\varphi(a_{j},b_{j})-\frac{\Disp (\pi_{j-1})-\Disp (\pi_{j})}{2}=\ief{a_{j}}{b_{j}}{\pi_{j-1}}.\]
By summing over all $j$, we find 
\begin{equation}
\diste (\pi)-\frac{1}{2}\Disp (\pi)=\sum_{j=1}^{\left|T^{*}\right|}\frac{1}{2}\ief{a_{j}}{b_{j}}{\pi_{j-1}},
\label{eq:ineff-decomp}
\end{equation}
which produces the desired result.
\myqed
\end{proof}
Note that by Lemma~\ref{lm:efficiency}, the right side of \eqref{eq:ineff-decomp} is always nonnegative.

Our algorithmic solution to the decomposition problem conceptually consists of two stages. In the first stage, a cycle is represented by a product of shorter adjacent cycles, each of which has the property that its support lies on a path in the Y-tree. It can be shown that the overall cost of the decomposition performed on each of the shorter cycles is minimized in this process. The second stage involves decomposing cycles that have supports that lie on paths. % In the following subsections, we focus on the first stage. The description and analysis of the second stage is postponed until the end of the section.

\subsection{Case 1: Cycles containing the central vertex}
\label{case1}
As before, denote the central vertex by $\cv$, and with slight abuse of notation, use \(B_{1}, B_{2}\  \text{and}\ B_{3}\) to denote both the three branches of the Y-tree and their corresponding vertex sets. Recall that the central vertex does not belong to any of the branches.

The decomposition procedure for this cycle type is described in Algorithm~\ref{alg:ctd}. 
The algorithm terminates when all subcycles $\pi_j$ have supports that lie on paths of the defining Y-tree.

\begin{algorithm}
\caption{central-td}
\SetKwFunction{this}{central-td}
\SetKwFunction{that}{path-td}
\SetKwFunction{concat}{Concatenate}
\SetKwFunction{branch}{Br}
%\tcc{Minimum weight transposition decomposition for cycles that contain the central vertex}
\label{alg:ctd}
\KwIn{A cycle $\pi=(\cv\,v_1\,\cdots\, v_{|\pi|-1})$ containing the central vertex $\cv$}
\KwOut{A minimum cost decomposition of $\pi$}
\BlankLine
\lIf{$\supp(\pi)$ is contained in a path of the Y-tree}{\Return \that($\pi$)}
$t\leftarrow\min\{{i\in [|\pi|-1]: v_{i}\in \text{\branch{$v_1$}} , v_{i+1}\notin \text{\branch{$v_1$}}\}}$\;
$\pi'\gets(\cv\,v_{t+1}\,\cdots\, v_{|\pi|-1})$\;
$\pi''\gets(\cv\,v_{1}\,\dotsc\, v_{t})$
\tcc*[r]{$\pi=\pi'\pi''$}
\Return \concat(\this($\pi'$), \that($\pi''$))
\end{algorithm}
\begin{lm}\label{lem:onpath}
Let $\pi_1$ and $\pi_2$ be two permutations such that $\supp(\pi_1)\cap \supp(\pi_2)=\{a\}$. The following are equivalent:
\begin{enumerate}
\item $\Disp (\pi_1\pi_2)=\Disp (\pi_1)+\Disp (\pi_2)$,
\item The vertex $a$ lies on the $(\pi_1(a)$,$\pi_2^{-1}(a))$-path.
\end{enumerate}
If the above conditions hold and, additionally, $\diste (\pi_1) =
\frac{1}{2}\Disp (\pi_1)$ and $\diste (\pi_2) = \frac{1}{2}\Disp
(\pi_2)$, then $\diste (\pi_1\pi_2) = \frac{1}{2}\Disp (\pi_1\pi_2)$.
\end{lm}
\begin{proof}  
  Since $\supp(\pi_1) \cap \supp(\pi_2) = \{a\}$, we have
  \[\Disp (\pi_1\pi_2)-\Disp (\pi_1)-\Disp (\pi_2) = \varphi(\pi_2^{-1}(a), \pi_1(a)) - \varphi(a, \pi_2^{-1}(a)) - \varphi(a, \pi_1(a)). \]
  Condition~1 holds if and only if the right side of this equation is $0$. Since $\varphi$ is strictly positive, the right side of
  this equation is $0$ if and only if Condition~2 holds.

  For the second part, Lemma~\ref{lem:metric-path} yields the lower bound $\diste(\pi_1\pi_2) \geq \frac{1}{2}\Disp(\pi_1\pi_2)$,
  while the hypotheses give the upper bound:
  \[ \diste(\pi_1\pi_2) \leq \diste(\pi_1) + \diste(\pi_2) = \frac{1}{2}\Disp(\pi_1) + \frac{1}{2}\Disp(\pi_2) = \frac{1}{2}\Disp(\pi_1\pi_2). \]
  \myqed
\end{proof}
\begin{lm}
The minimum decomposition cost of a cycle \(\pi \) containing the central vertex equals one half of its displacement, i.e.,
\[
\diste (\pi)=\frac {1}{2}\Disp (\pi).
\]
\label{lem:central}
\end{lm}
\begin{proof}
We use induction on $|\pi|$. The smallest (non-trivial)
cycle $\pi$ that contains the central vertex $\cv$ is of the form
$(\cv\, b)$ for some $b$ and has only two vertices. Thus 
\[
\diste (\pi)=\varphi(\cv,b)=\frac{1}{2}\Disp (\pi).
\]

Now suppose for any cycle of size at most $m-1$, the lemma holds.
We show that it also holds for a cycle $\pi$ of size $m$. Algorithm~\ref{alg:ctd} finds two cycles $\pi'$ and $\pi''$ such that $\pi=\pi'\pi''$, $\supp(\pi') \cap \supp(\pi'') = \{\cv\}$, and the cycle $\pi''$ lies on
a path of the Y-tree.

Since $v_t$ and $v_{t+1}$ lie on different branches, $\cv$ lies on
the unique $v_t,v_{t+1}$-path. Since $\pi'$ has size at most $m-1$,
the induction hypothesis yields $\diste (\pi')=\frac{1}{2}\Disp
(\pi')$, while Lemma~\ref{lem:metric-path} yields $\diste
(\pi'')=\frac{1}{2}\Disp (\pi'')$. By Lemma~\ref{lem:onpath}, we
conclude that $\diste(\pi) = \frac{1}{2}\Disp(\pi)$.\myqed
\end{proof}

\begin{figure}
\centering
\mbox{
\subfigure[]{\includegraphics[width=6.5cm]{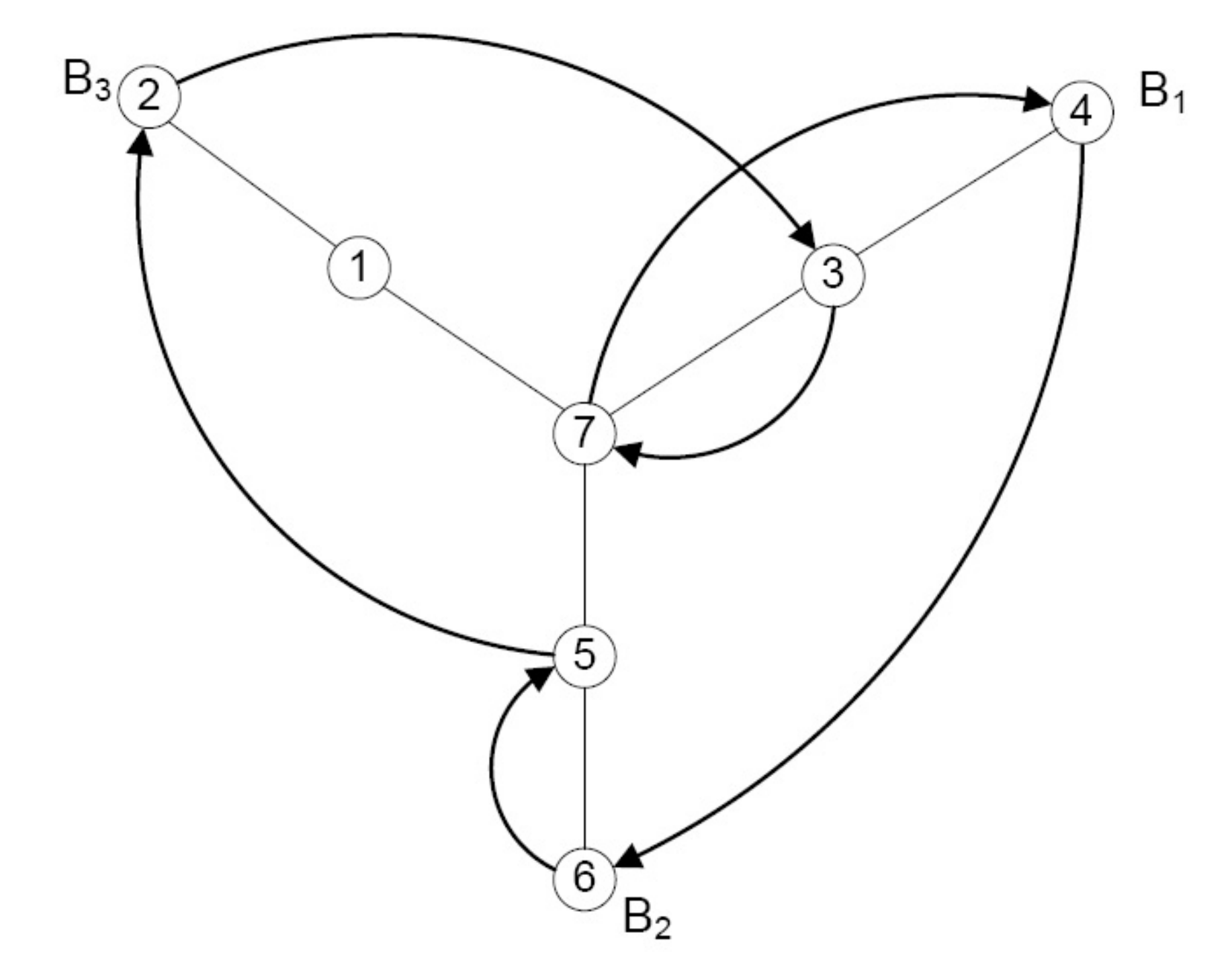}}\quad
\subfigure[]{\includegraphics[width=6.5cm]{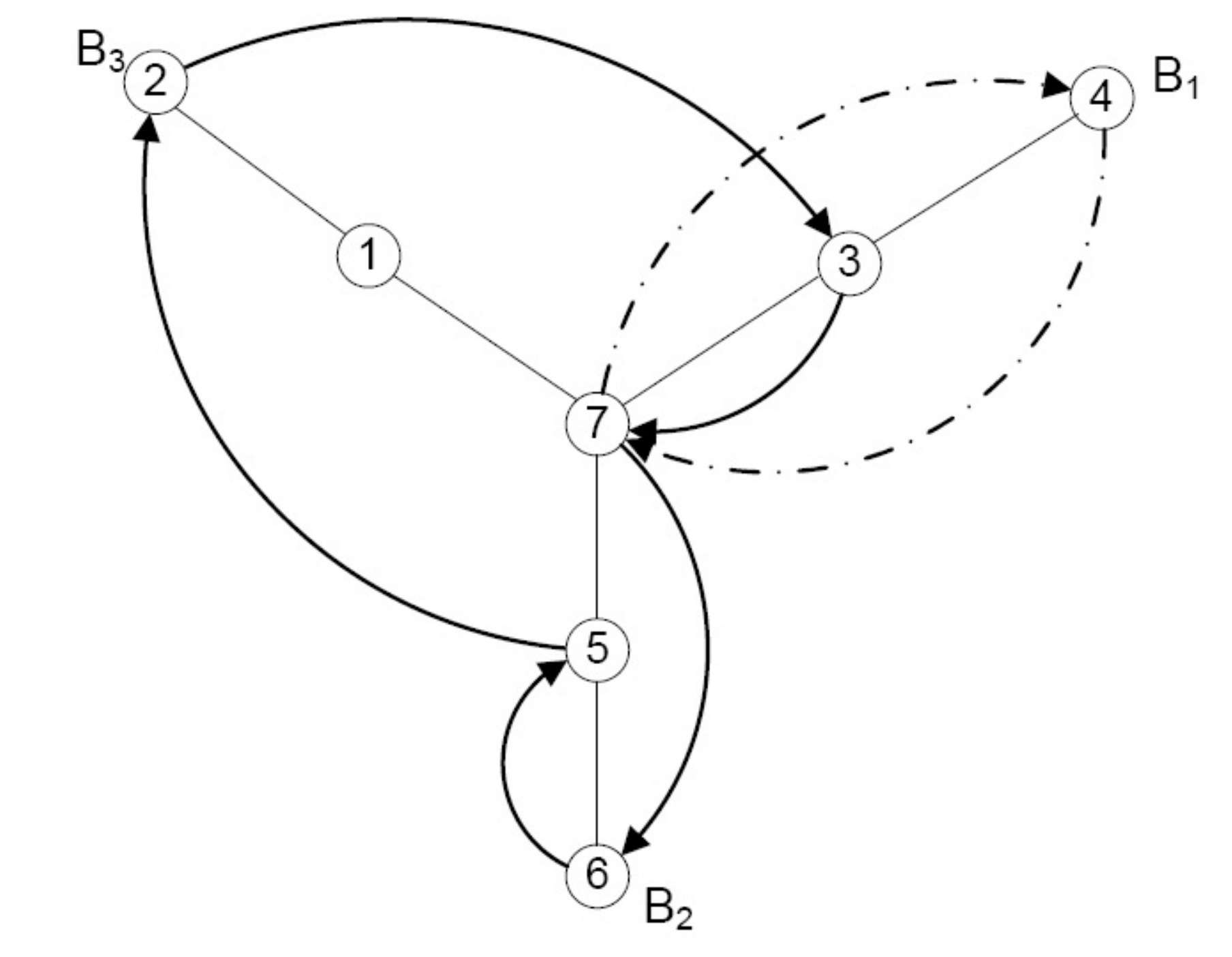}}\quad
}
\caption{1\,) A Y-tree and input cycle (7\,4\,6\,5\,2\,3). After the first iteration of Algorithm~\ref{alg:ctd}, the arc $(4 \to 6)$ is replaced by two arcs, decomposing the original
cycle (7\,4\,6\,5\,2\,3) into a product of two adjacent cycles, i.e., (7\,4\,6\,5\,2\,3)=(7\,6\,5\,2\,3)(7\,4), as shown in 2\,).}
\label{fig:balancedcycle}
\end{figure}

\subsection{Case 2: Balanced cycles}
\label{case2}
Given that a cycle containing the central vertex was analyzed in Case 1 of our exposition, we henceforth tacitly assume that the balanced cycles considered in this section do not contain the central vertex.

\begin{lm}\label{lem:balanced}
For a balanced cycle $\pi$, the minimum decomposition cost of $\pi$ equals one half of its displacement, i.e.,
\[
\diste (\pi)=\frac {1}{2}\Disp (\pi).
\]
\end{lm}
\begin{proof}
We prove the lemma by induction on $|\pi|$.

Clearly, the lemma holds for $|\pi|=2$. We therefore assume next that it also holds for all cycles $\kappa$ with $|\kappa|<m$. We then show that the claimed result also holds for $\pi$, where $\pi$ is an arbitrary cycle such that $|\pi|=m$. 

Let $\pi=(v_1\dotsc v_m)$, and without loss of generality, let the support of $\pi$ span all three branches of the Y-tree (if the support of the cycle were to lie on two branches only, the desired result would immediately follow from \eqref{path-case}). Consider two distinct indices $t$ and $l$ modulo $m$, such that $v_{t}$ and $v_{l+1}$ belong to the same branch, say $B_1$, with $v_{t+1},\dotsc,v_{l}$ belonging to a different branch, say $B_2$. Such indices $t$ and $l$ must exist since the cycle $\pi$ is balanced.

We consider two cases, depending on which one of the two vertices $v_t$ and $v_{l+1}$ is closer to the center $\cv$. First, suppose $v_t$ is closer to $\cv$, that is, $\varphi(\cv,v_t)<\varphi(\cv,v_{l+1})$. An illustrative example is shown in Fig.~\ref{fig:case2a}.1. 
Let $\pi'=(v_1 \dotsm v_{t-1}v_{t}v_{l+1}v_{l+2}\dotsm v_m)$ and $\pi''=(v_t\dotsm v_l)$. Note that $\pi=\pi'\pi''$, that $\supp(\pi') \cap \supp(\pi'') = \{v_t\}$, and that $\pi'$ is balanced while $\pi''$ lies on a path (See Fig.~\ref{fig:case2a}.2). The induction hypothesis yields $\diste (\pi') = \frac{1}{2}\Disp(\pi')$, while Lemma~\ref{lem:central} yields $\diste (\pi'') = \frac{1}{2}\Disp(\pi'')$. Since $v_t$ lies on the $v_l,v_{l+1}$-path, Lemma~\ref{lem:onpath} yields $\diste(\pi) = \frac{1}{2}\Disp(\pi)$.

Next, suppose that $\varphi(\cv,v_t)>\varphi(\cv,v_{l+1})$, as
illustrated in Fig.~\ref{fig:case2b}. In this case, let $\pi'=(v_1 \dotsm
v_{t-1}v_{t}v_{l+1}v_{l+2}\dotsm v_m)$ and let $\pi''=(v_{t+1}\dotsm
v_{l+1})$. Now $\pi=\pi''\pi'$, with $\pi'$ lying on a path and with
$\pi''$ balanced, so we again have $\diste(\pi') =
\frac{1}{2}\Disp(\pi')$ and $\diste(\pi'') =
\frac{1}{2}\Disp(\pi'')$. Since $\supp(\pi') \cap \supp(\pi'') =
\{v_{l+1}\}$ and $v_{l+1}$ lies on the $v_t,v_{t+1}$-path,
Lemma~\ref{lem:onpath} again yields $\diste(\pi) =
\frac{1}{2}\Disp(\pi)$.  \myqed\end{proof}

\begin{figure}
\centering
\mbox{
\subfigure[]{\includegraphics[width=6.5cm]{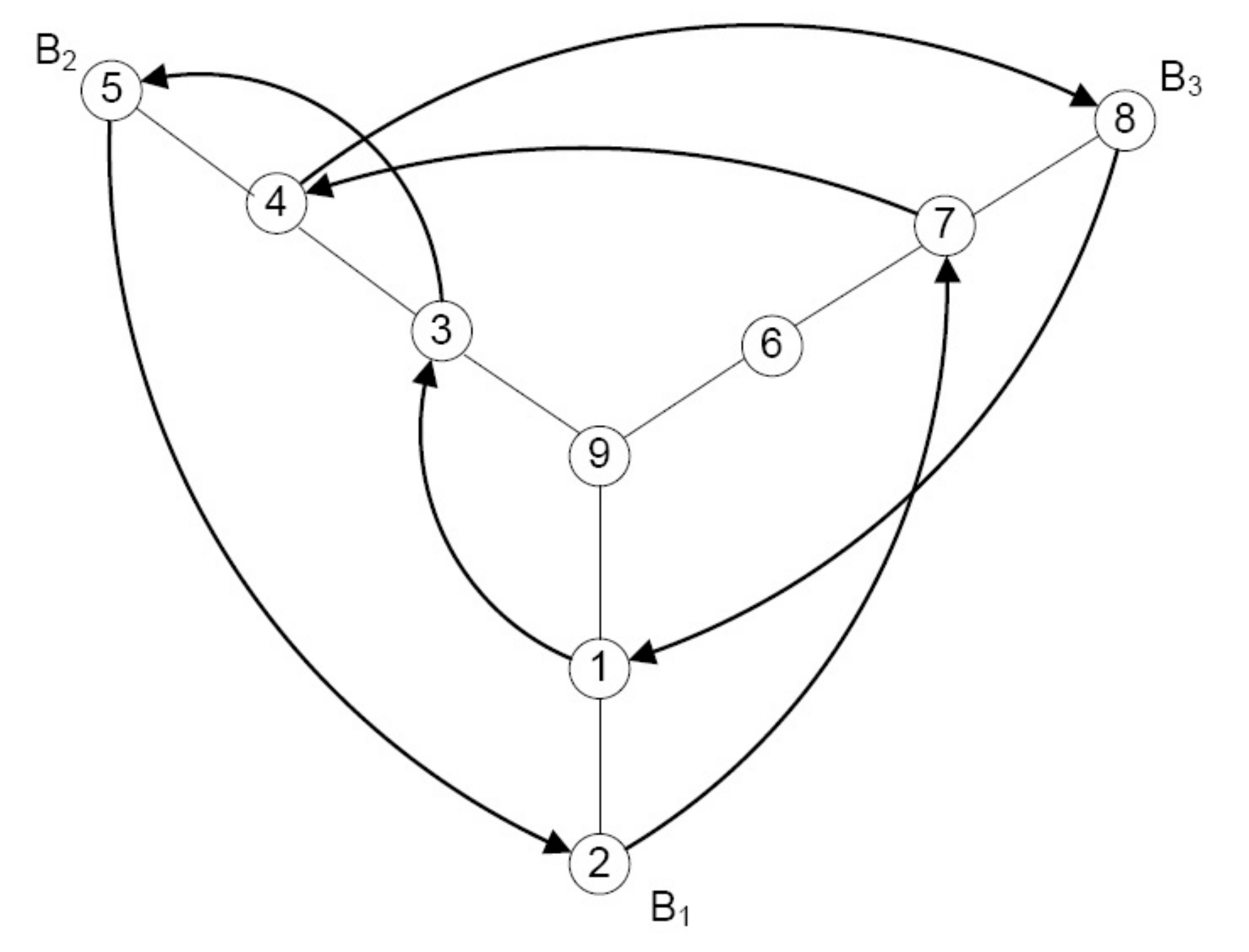}}\quad
\subfigure[]{\includegraphics[width=6.5cm]{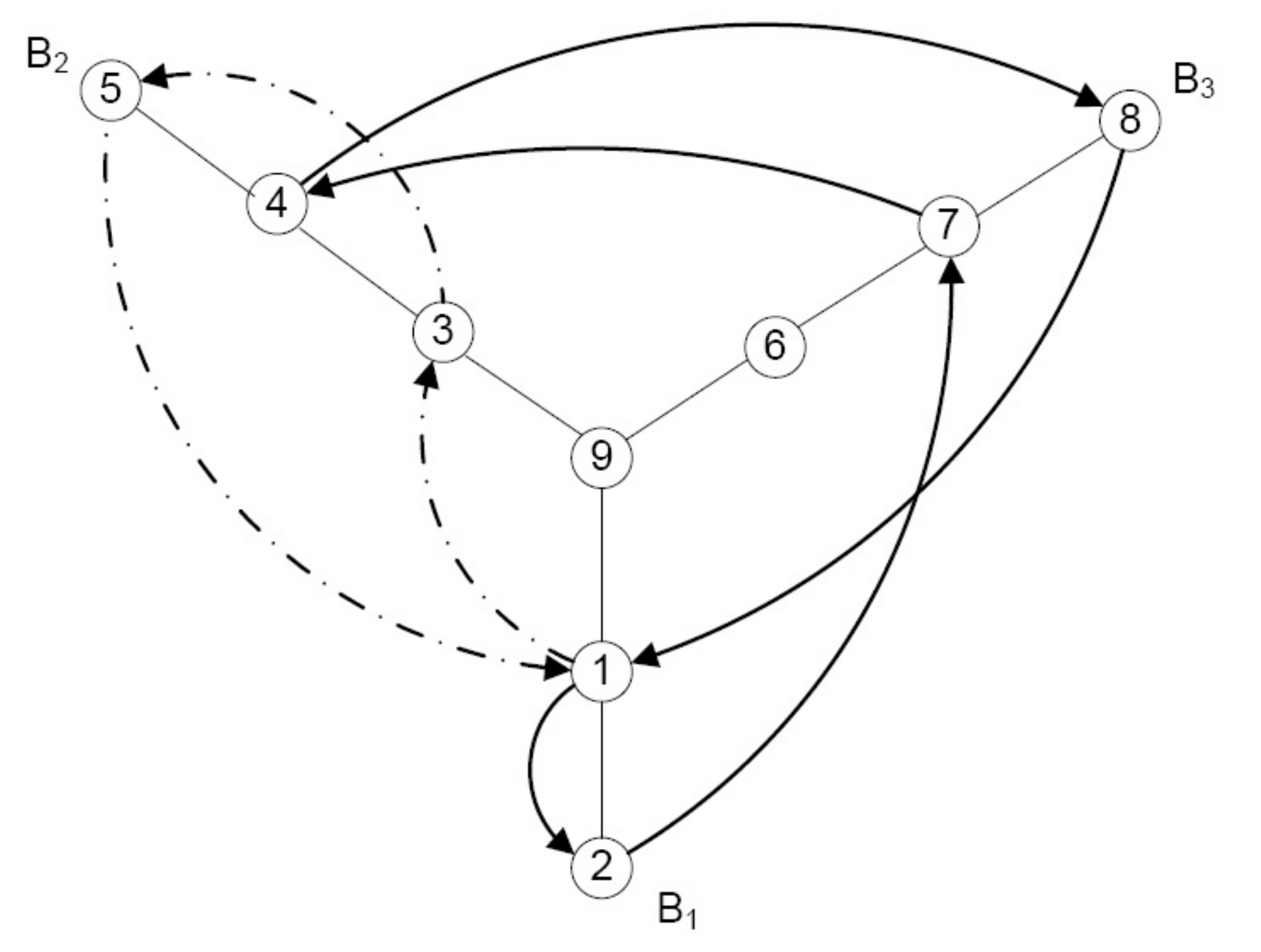}}\quad
}
\caption{An example of the decomposition procedure when $\varphi(\cv, v_t)<\varphi(\cv, v_{l+1})$. The cycle equals $\pi=(1\,3\,5\,2\,7\,4\,8)$, with $v_1=1\in B_1$. The first visited arc between branches is $(1 \to 3)$, i.e., $v_t=1\in B_1,\, v_{t+1}=3\in B_2$; the second visited arc between branches is $(5 \to 2)$, i.e., $v_l=5\in B_2, \, v_{l+1}=2\in B_1$. As $v_{l+1}=2\in B_1$, we decompose (1\,3\,5\,2\,7\,4\,8) into two shorter cycles (1\,2\,7\,4\,8) and (1\,3\,5), i.e., (1\,3\,5\,2\,7\,4\,8)=(1\,2\,7\,4\,8)(1\,3\,5).}
\label{fig:case2a}
\end{figure}

\begin{figure}
\centering
\mbox{
\subfigure[]{\includegraphics[width=6.5cm]{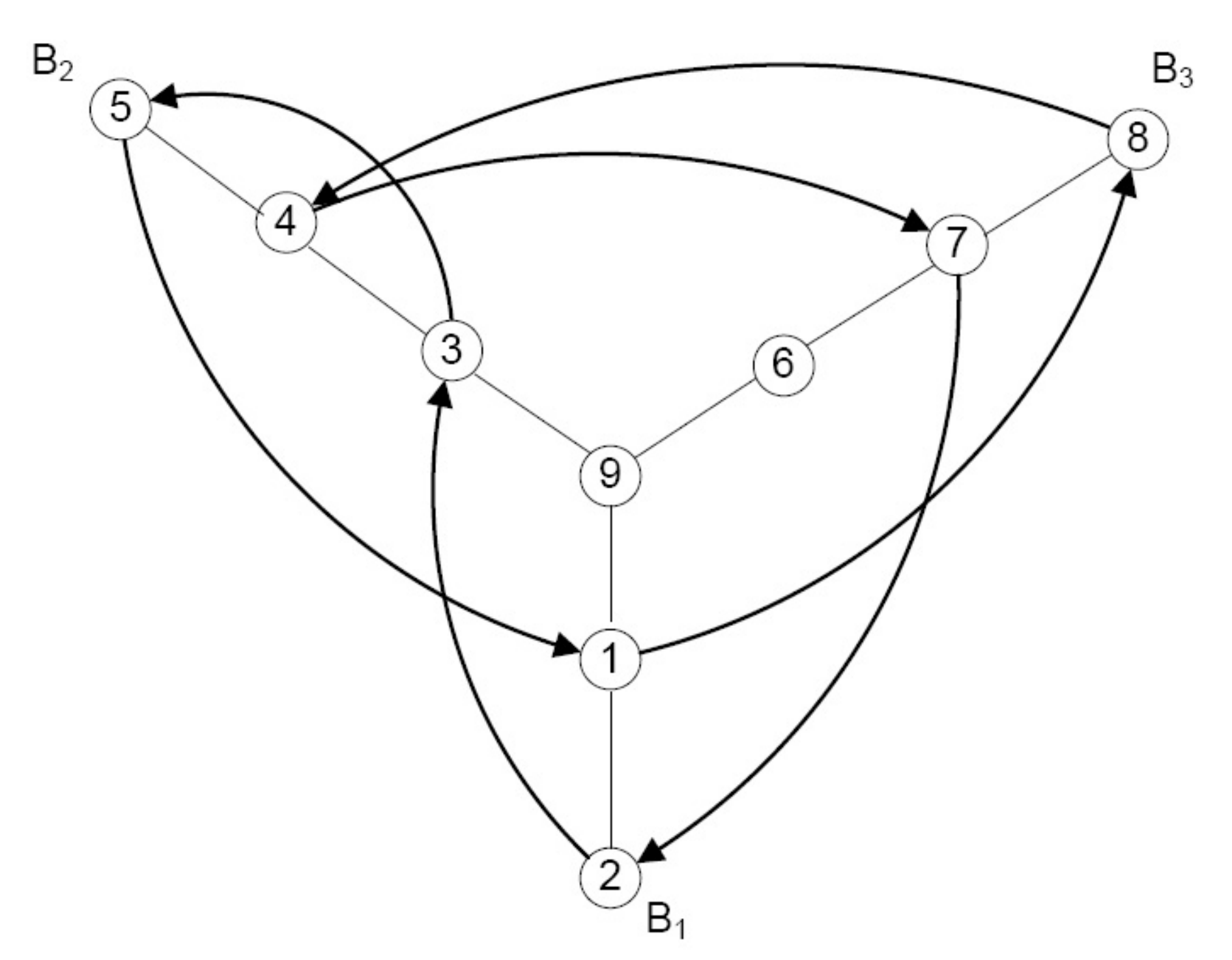}}\quad
\subfigure[]{\includegraphics[width=6.5cm]{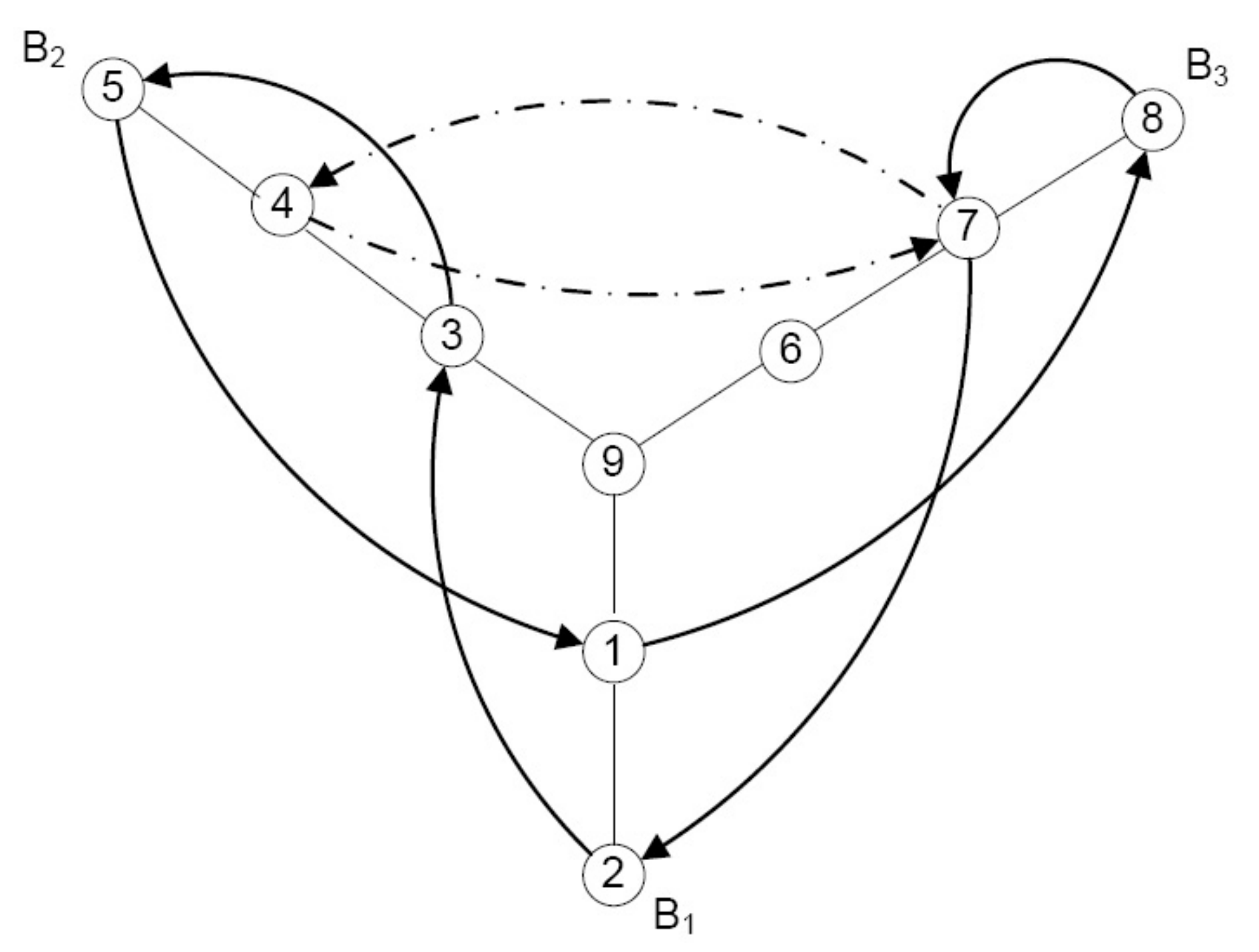}}\quad
}
\caption{An example of the decomposition procedure when $\varphi(\cv, v_t)>\varphi(\cv, v_{l+1})$. The cycle equals $\pi=(1\,8\,4\,7\,2\,3\,5)$, with $v_1=1\in B_1$. The first visited arc between branches is $(1 \to 8)$, i.e., $v_t=1\in B_1,\, v_{t+1}=8\in B_3$; the second visited arc between branches is $(8 \to 4)$, i.e., $v_l=8\in B_3,\, v_{l+1}=4\in B_2$. As $v_{l+1}=4\notin B_1$, we add $(8 \to 4)$ to the stack $S$ and move on to the arc $(4 \to 7)$. Since $v_{l+1}=7\in B_3$, we decompose (1\,8\,4\,7\,2\,3\,5) into two shorter cycles (4\,7) and (1\,8\,7\,2\,3\,5), i.e., (1\,8\,4\,7\,2\,3\,5)=(4\,7)(1\,8\,7\,2\,3\,5).}
\label{fig:case2b}
\end{figure}

Based on the proof of Lemma~\ref{lem:balanced}, we present Algorithm~\ref{alg:btd}, which describes the steps for finding a minimum cost decomposition for a balanced cycle $\pi$. We use a push-down stack data structure, with the standard push, pop, and peek operations, to search for indices $t$ and $l$ with the properties described in the proof of the above lemma. The stack is denoted by $S$. 

\begin{algorithm}
\caption{balanced-td}
%\tcc{Adjacent cycle decomposition for balanced cycles}

\label{alg:btd}
\SetKwFunction{this}{balanced-td}
\SetKwFunction{that}{path-td}
\SetKwFunction{push}{push}
\SetKwFunction{pop}{pop}
\SetKwFunction{peek}{peek}
\SetKwFunction{branch}{Br}
\SetKwFunction{concat}{Concatenate}
\KwIn{A balanced cycle $\pi=(v_1\,v_2\,\cdots\, v_{|\pi|})$}
\KwOut{A minimum cost decomposition of $\pi$}%An adjacent cycle decomposition $T=(\pi_1,\dotsc,\pi_{{|T|}})$ of $\pi$}%, where for all $j$, $\supp(\pi_{j})$ is contained in a path of the Y-tree}
\BlankLine
%\lIf{$\supp(\pi)$ is contained in a path}{\Return}
$S\gets\emptyset$, $a\gets v_1$
\tcc*[r]{$a$ denotes the last vertex visited}
$T\gets\pi$

\tcc{$T$ will be an adjacent cycle decomposition $T=\pi_{1},\dotsc, \pi_{k}$ of $\pi$ (i.e. $\pi=\pi_1\dotsm\pi_k$), where $\supp(\pi_j),\, j\in[k]$ is contained in a path of the Y-tree}

\While{$\pi$ spans all three branches}
{
	$j\gets \min\{i:\pi^{i+1}(a)\notin \text{\branch{$a$}}\}$\;
	$c_1\gets \pi^j(a)$, $c_2\gets \pi(c_1)$
	\tcc*[r]{ $c_1 \to c_2$ is the first unvisited arc that leaves \branch{$a$}}
	$a\gets c_2$\tcc*[r]{update last visited vertex}
	
	\eIf{$S=\emptyset$}
	{
		\push{$S,c_1 \to c_2$}\;
	}{
		$(b_1 \to b_2) \gets$\peek{$S$}\;
		\eIf{$c_2\in\text{\branch{$b_1$}}$}
		{
			\eIf{$\varphi(\cv, b_1)<\varphi(\cv, c_2)$}
			{
				$\pi''\gets(b_1\,\cdots \, c_1)$ and $\pi^{'}\gets \pi(\pi'')^{-1}$\;
				In $T$, replace $\pi$ with the pair $\pi',\pi''$
				\tcc*[r]{the cycle $\pi$ is decomposed into two cycles}
			}{
				$\pi''\gets(b_2\,\cdots \, c_2)$ and $\pi^{'}\gets (\pi'')^{-1}\pi$\;
				In $T$, replace $\pi$ with the pair $\pi'',\pi'$\;
			}
			\pop{$S$}\;
			$\pi\gets \pi^{'}$\;
		}{
		\push{$S, c_1 \to c_2$}\;
		}
	}}
	$U\gets ()$\;
	\For{$\kappa\in T$}{
		$U\gets$\concat($U$, \that($\kappa$))\;
	}
	\Return $U$\;
\end{algorithm}

We follow the closed walk induced by $\pi$, starting from an arbitrary vertex\footnote{Although the procedure works for an arbitrarily chosen vertex, for ease of demonstration, in Algorithm~\ref{alg:btd}, we simply fix the initial vertex.} in the support of the cycle until encountering a branch-changing arc. Such an arc is pushed into the stack $S$. We keep following the closed walk while pushing arcs in or out of the stack $S$. Only branch-changing arcs may be added to the stack. Once a branch-changing arc in the ``opposite branch direction'' of the arc at the top of the stack is encountered, the two arcs are paired up and removed from the stack. The paired arcs dictate the choice of the transpositions $(v_t,v_{t+1})$ and $(v_l,v_{l+1})$ in the proof of the previous result, and are used to decompose the current cycle. The procedure is repeated until all the vertices of the cycle are visited exactly once. As each vertex is visited once, the running time of the algorithm is linear in $|\pi|$.

\subsection{Case 3: Unbalanced Cycles}
\label{case3}
In this section, we will determine $\diste(\pi)$ in the case where $\pi$ is an unbalanced
cycle. The proof relies on a lower bound for $\diste(\pi)$ for general permutations $\pi$,
which we show to hold with equality when $\pi$ is an unbalanced cycle. To prove this lower
bound, we show that every permutation has a min-cost sorting with a particularly useful
technical property. We first prove a few smaller lemmas.
\begin{lm}\label{lem:commuteself}
  Let $\alpha$ and $\beta$ be permutations.
  If $\supp(\alpha) \cap \supp(\beta\alpha) = \emptyset$, then
  $\alpha\beta = \beta\alpha$.
\end{lm}
\begin{proof}
  Since $\beta\alpha(x) = x$ for all $x \in \supp(\alpha)$, we
  see that $\beta = \alpha^{-1}$ on $\supp(\alpha)$. Thus,
  we can write $\beta = \gamma\alpha^{-1}$ where $\supp(\gamma) \cap \supp(\alpha) = \emptyset$.
  Since $\alpha$ commutes with both $\gamma$ and $\alpha^{-1}$, we see that
  $\alpha$ commutes with $\beta$. \myqed
\end{proof}
\begin{lm}\label{lem:nocycling}
  Let $\sigma_1, \ldots, \sigma_p$ be transpositions, and let $\rho =
  \sigma_1\cdots\sigma_p$.  If for all $i \in [p]$ we have $\supp(\sigma_i)
  \cap \supp(\sigma_i\sigma_{i+1}\sigma_{i+2}\cdots\sigma_{i+(p-1)}) =
  \emptyset$, where subscripts are taken modulo $p$, then $\rho$ is the
  identity permutation.
\end{lm}
\begin{proof}
  Repeatedly applying Lemma~\ref{lem:commuteself} shows that
  for all $i$,
  \[ \rho = \sigma_i\sigma_{i+1}\cdots\sigma_{i+p-1}, \]
  again with subscripts modulo $p$.
  Thus, by hypothesis, we have $\supp(\rho) \cap \supp(\sigma_i) =
  \emptyset$ for all $i$.  Since $\supp(\rho) \subset \supp(\sigma_1)
  \cup \cdots \cup \supp(\sigma_p)$, it follows that $\supp(\rho) =
  \emptyset$.  \myqed
\end{proof}
We now state and prove our the min-cost sorting lemma.
\begin{lm}\label{lem:supp}
  Every permutation $\pi$ has a minimum-cost sorting $(\tau_1, \ldots, \tau_k)$ such that
  for all $i$,
  \[ \supp(\tau_i) \cap \supp(\tau_i\tau_{i+1}\cdots\tau_n) \neq \emptyset. \]
\end{lm}
\begin{proof}
  Let $(\tau_1, \ldots, \tau_k)$ be a minimum-cost sorting of $\pi$.
  We will show that there is some $\eta \in S_k$ such that
  $(\tau_{\eta(1)}, \tau_{\eta(k)})$ is a sorting of $\pi$
  with the desired support property. Clearly, any such sorting
  is also a minimum-cost sorting.

  We define an algorithm to manipulate the sorting, and write $\tau_i$
  to refer to the transposition \emph{currently} in the $i$th position
  of the sorting. Say that a transposition $\tau_i$ is \emph{bad} in
  the current sorting if $\supp(\tau_i) \cap
  \supp(\tau_i\cdots\tau_{k}) = \emptyset$. Our goal is to permute the
  transpositions so that there is no bad transposition. Consider the
  following algorithm:
  \begin{algorithm}
    \caption{fix-sorting}
    \label{alg:fix}
    \KwIn{A min-cost sorting $(\tau_1, \ldots, \tau_{k})$ of $\pi$.}
    \KwOut{A min-cost sorting of $\pi$ with no bad transpositions.}
    \While{there is a bad transposition}
    {
      $B \gets \min\{i:\text{$\tau_i$ is bad}\}$\;
      Move $\tau_B$ to the end of the sorting, keeping all other transpositions in the same relative order;
    }
  \end{algorithm}

  By Lemma~\ref{lem:commuteself}, if $\tau_B$ is bad before we execute
  Step~3, then
  \[ \tau_B\tau_{B+1}\cdots\tau_k = \tau_{B+1}\cdots\tau_{k}\tau_B,\]
  so the product $\tau_1\cdots\tau_k$ does not change after executing Step~2. Thus, at all times
  $(\tau_1, \cdots, \tau_k)$ is a minimum-cost sorting of $\pi$. If the algorithm
  terminates, then it yields a minimum-cost sorting of $\pi$ with no bad
  transpositions, as desired. We now show that the algorithm terminates.

  Let $B_i$ denote the index $B$ chosen in Step~2 on the $i$th
  iteration of the algorithm.  The key observation is that the
  sequence $B_1, B_2, \ldots$ is nondecreasing: if $B$ is the index of
  the leftmost bad transposition and $C < B$, then by
  Lemma~\ref{lem:commuteself}, the rearrangement in Step~2 does not
  alter the product $\tau_{B}\cdots\tau_k$, so $\tau_{C}$
  remains good after Step~2. Thus, if the algorithm does not
  terminate, then the sequence $B_i$ is eventually constant: the
  algorithm is repeatedly choosing the same index $B$. Let $(\tau_1, \ldots, \tau_k)$
  be the current sorting when we first choose the index $B$, and let $\sigma_1, \ldots, \sigma_p = \tau_B, \ldots, \tau_k$.
  Since the index $B$ is bad for the rest of the algorithm's run, we have \[\supp(\sigma_i) \cap
  \supp(\sigma_i\sigma_{i+1}\cdots\sigma_{i+(p-1)}) = \emptyset\]
  for all $i \in [p]$, where indices are taken modulo $p$.  By Lemma~\ref{lem:nocycling}, it follows that
  $\sigma_1\cdots\sigma_p = e$. Now $(\tau_{1}, \ldots, \tau_{B-1})$ is a
  lower-cost sorting of $\pi$, contradicting the assumption that
  $(\tau_1, \ldots, \tau_k)$ is a minimum-cost sorting. \myqed
\end{proof}
We are now ready to prove the main result of the section.
\begin{lm}
For an unbalanced permutation $\pi$, we have 
\begin{equation}\label{eq:whatever5}
\diste (\pi)\ge\frac {1}{2}\Disp (\pi)+ \min_{v_{i}\in \supp(\pi)}  \varphi(\cv,v_{i}),
\end{equation}
Furthermore, if $\pi$ has only one non-trivial cycle, the inequality is satisfied with equality.
\label{lem:unbalanced}
\end{lm}
\begin{proof}
To prove Lemma \ref{lem:unbalanced}, we first derive the lower bound~\eqref{eq:whatever5}, which we subsequently show in a constructive manner to be achievable. Intuitively, the bound suggests that one should
first merge the central vertex into the cycle via a smallest cost transposition and then decompose the newly formed cycle. 
Despite the apparent simplicity of the claim, the proof of the result is rather technical.

Let $T^*=(\tau_{1}, \cdots,\tau_{|T^*|})$ be a minimum cost sorting of $\pi$
satisfying the conclusion of Lemma~\ref{lem:supp}. Define
$\pi_j=\pi_{j-1}\tau_j$, for all $1\le j\le |T^*|$, with $\pi_0=\pi$.
For all $j$ in $\{1, \ldots, |T^*|\}$, we have $\pi_{j-1} = \tau_{|T^*|}\cdots\tau_{j}$,
so by the choice of $\pi$, we have $\supp(\tau_j) \cap \supp(\pi_{j-1}) \neq \emptyset$. Finally, let
\[
f_j = \frac12\sum_{i=1}^j\ief* {\tau_i} {\pi_{i-1}} + \centerineff(\pi_j),
\]
where
\[
\centerineff(\sigma)=\begin{cases}
0, & \quad\mbox{if }\sigma\mbox{ is balanced}\\
\min_{v\in\supp(\sigma)}\varphi(\cv,v), & \quad\mbox{else}
\end{cases}  
\]
for any permutation $\sigma$.
Below, we show that $f_j$ is non-decreasing, implying that
\[
\min_{v\in\supp(\pi)}\varphi(\cv,v)=f_{0}\le f_{|T^{*}|}= \frac12\sum_{i=1}^{|T^*|}\ief* {\tau_i} {\pi_{i-1}}=\diste (\pi)-\frac{1}{2}\Disp (\pi),
\]
where the last equality follows from Lemma~\ref{lm:lm1-improv}. This proves~\eqref{eq:whatever5}.

To show that $f_j$ is non-decreasing, it suffices to show that 
\begin{equation}\label{eq:we6}
\frac12 \ief* {\tau_j} {\pi_{j-1}} \ge \centerineff(\pi_{j-1}) - \centerineff(\pi_{j}).
\end{equation}
If $\pi_{j-1}$ is balanced or $\cv\in\supp(\pi_{j-1})$, the right side of \eqref{eq:we6} is non-positive and so~\eqref{eq:we6} holds trivially. Hence, we assume $\pi_{j-1}$ is unbalanced. There are three cases to consider for $\pi_j$: balanced with $\cv\not\in\supp(\pi_{j})$; $\cv\in\supp(\pi_{j})$; and unbalanced. We prove \eqref{eq:we6} for each case separately:

\paragraph{$\pi_j$ is balanced with $\cv\not\in\supp(\pi_{j})$.} 
In this case, since $\centerineff(\pi_j)=0$, we must show 
\[\frac12 \ief* {\tau_j} {\pi_{j-1}} \ge \min_{v_{i}\in \supp(\pi_{j-1})}  \varphi(\cv,v_{i})\] 
The transposition $\tau_{j}=(a\,b)$ in $T^*$ changes the balance of arcs between two branches. In other words, for some $i,k$, we have $l_{ik}^{\pi_{j-1}}-l_{ki}^{\pi_{j-1}}\neq l_{ik}^{\pi_{j}}-l_{ki}^{\pi_{j}}$. 
Since the balance of arcs changes, one cannot encounter any of the following placements of the vertices $a$, $b$, $a'=\pi_{j-1}(a)$, and $b'=\pi_{j-1}(b)$ on the branches of the Y-tree:
\vspace{-.1cm}
\begin{itemize}
\item $\br[a]=\br[b]$;
\item $\br[\pi_{j-1}(a)]=\br[\pi_{j-1}(b)]$;
\item $\br[a]=\br[\pi_{j-1}(a)]$ and $\br[b]=\br[\pi_{j-1}(b)]$;
\item $\br[a]=\br[\pi_{j-1}(b)]$ and $\br[b]=\br[\pi_{j-1}(a)]$.
\end{itemize}

Since $\br[a]\neq\br[\pi_{j-1}(a)]$ or $\br[b]\neq\br[\pi_{j-1}(b)]$, by symmetry, we may assume $\br[a]\neq\br[\pi_{j-1}(a)]$. The cases not covered in the previous list satisfying $\br[a]\neq\br[\pi_{j-1}(a)]$ are shown in Fig.~\ref{fig:balance-changing}. Note that in the figure, the exact ordering of vertices on the same branch is irrelevant. 

\begin{figure}
\centering
%\mbox{
\subfigure[]{\includegraphics[height=4cm]{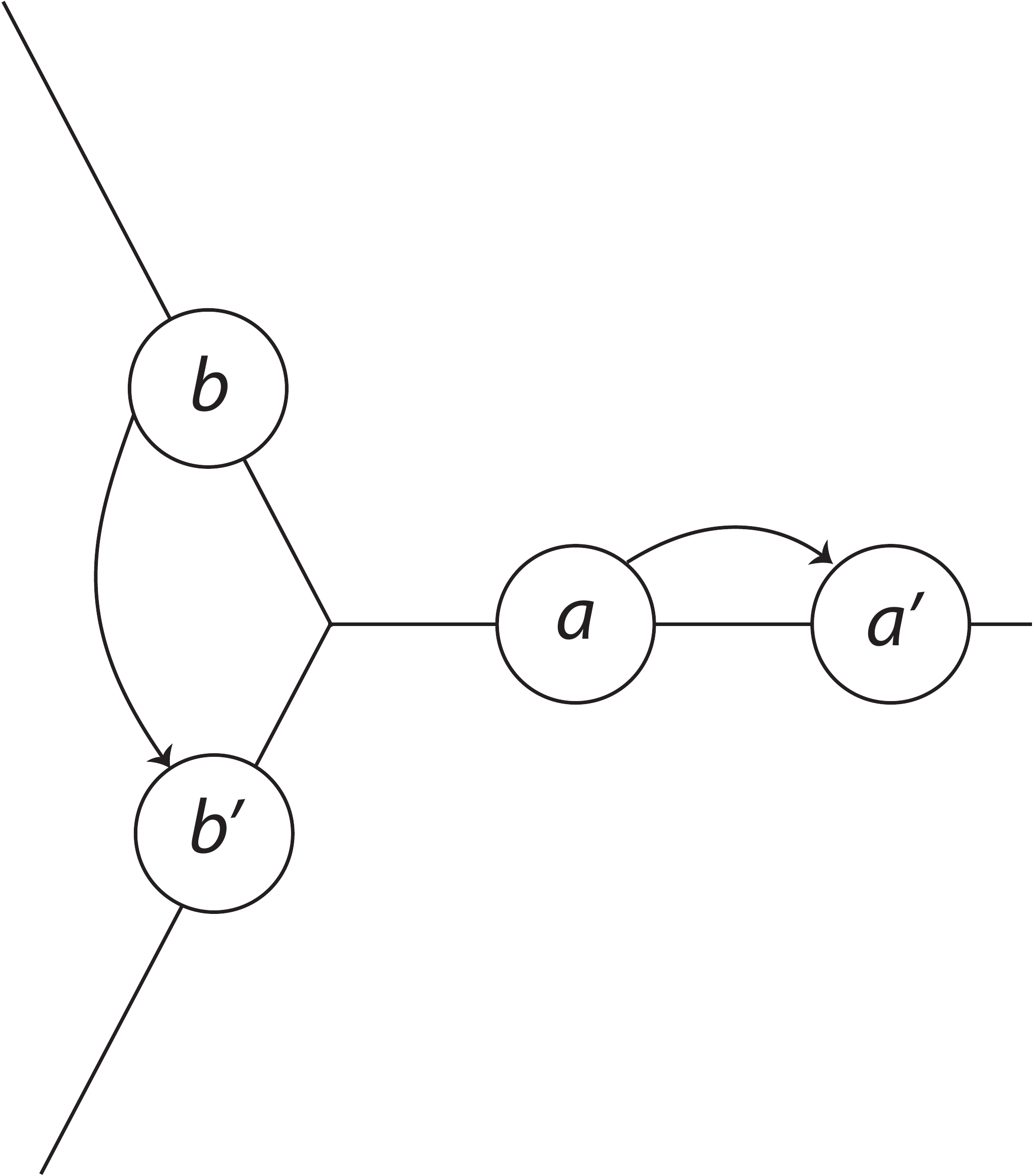}}\quad
\subfigure[]{\includegraphics[height=4cm]{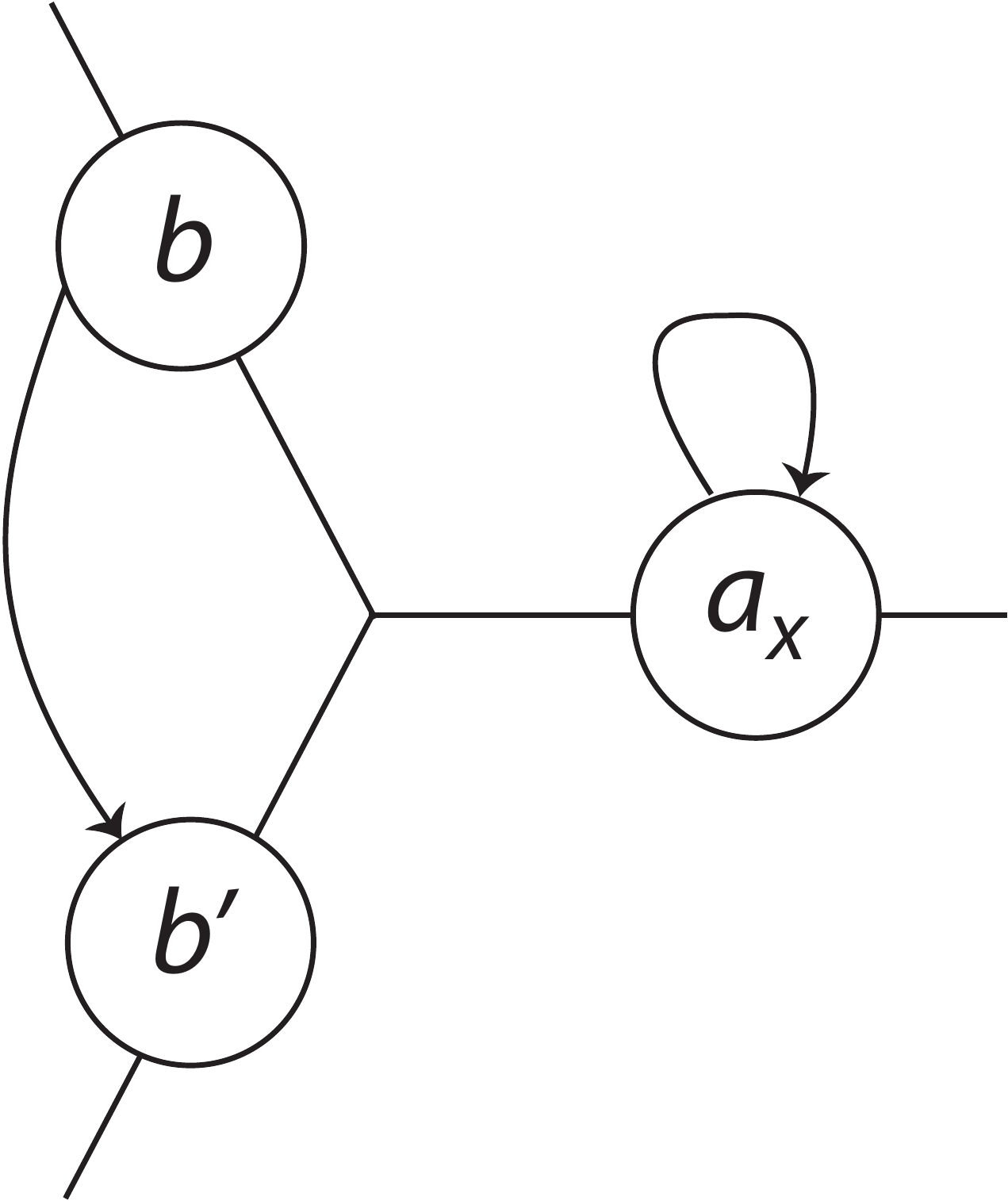}}\quad
\subfigure[]{\includegraphics[height=4cm]{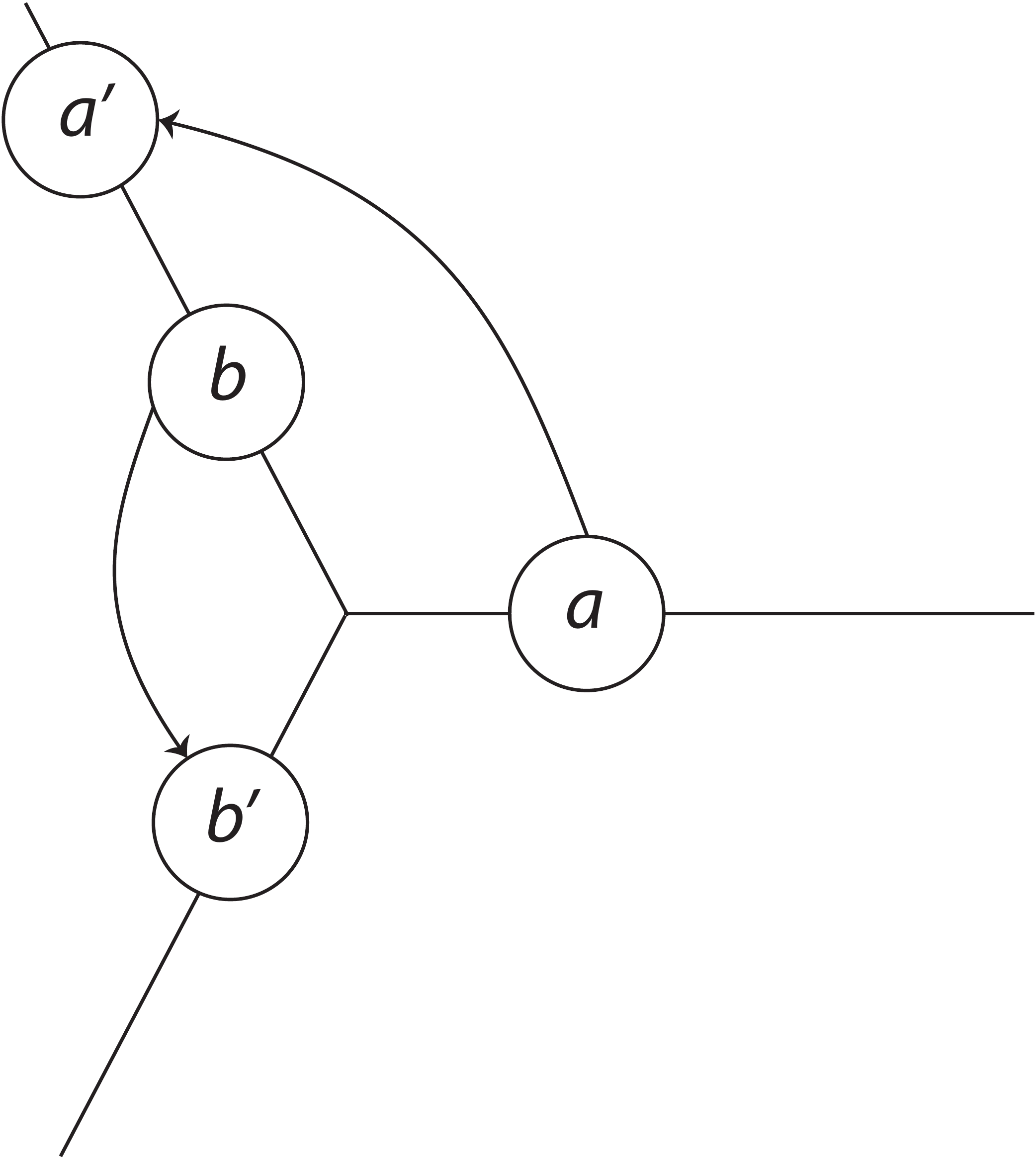}}
\subfigure[]{\includegraphics[height=4cm]{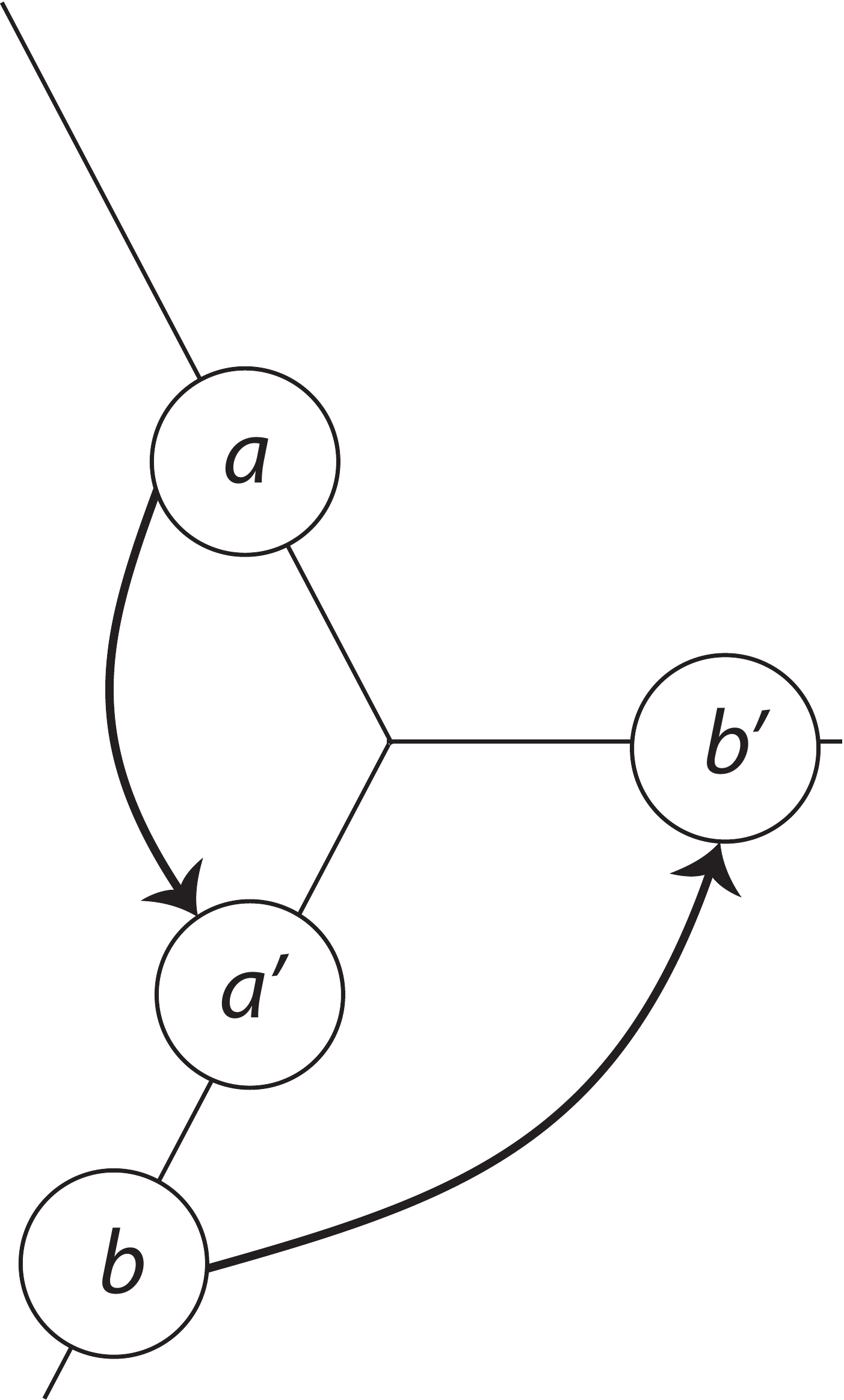}}
%}
\caption{Illustration for the proof of Lemma~\ref{lem:unbalanced}. Depicted are the configurations for $a$, $b$, $a'=\pi_{j-1}(a)$ and $b'=\pi_{j-1}(b)$ that change the balance between branches.}
\label{fig:balance-changing}
\end{figure}

For cases 1), 2) and 3), we have
\[\varphi(a,b)+\varphi(b,\pi_{j-1}(a))-\varphi(a,\pi_{j-1}(a))=2\varphi(b,\cv)\ge 2\min_{v_{i}\in \supp(\pi_{j-1})}  \varphi(\cv,v_{i}),\]
%where we used $b'=\pi_{j-1}(b)$, $a'=\pi_{j-1}(a)$ to simplify notation. 
and for case 4), we have
\[\varphi(a,b)+\varphi(a,\pi_{j-1}(b))-\varphi(b,\pi_{j-1}(b))=2\varphi(a,\cv) \ge 2\min_{v_{i}\in \supp(\pi_{j-1})}  \varphi(\cv,v_{i}).\]
Hence, by \eqref{eq:ineff-nonneg}, it follows that
\[\frac12\ief a b {\pi_{j-1}}\ge \min_{v_{i}\in \supp(\pi_{j-1})}  \varphi(\cv,v_{i}).\]

\paragraph{$\pi_j$ contains the central vertex, i.e., $\cv\in\supp(\pi_{j})$.}
In this case, since $\centerineff(\pi_j)=0$, we must show 
\[\frac12 \ief* {\tau_j} {\pi_{j-1}} \ge \min_{v_{i}\in \supp(\pi_{j-1})}  \varphi(\cv,v_{i})\] 
Since $\pi_{j-1}$ is unbalanced, it does not contain the central vertex. Since $\supp(\pi_j) \subset \supp(\pi_{j-1}) \cup \supp(\tau_j)$, this implies $\cv \in \supp(\tau_j)$. Write $\tau_j=(\cv\,b)$. Since $\supp(\tau_j) \cap \supp(\pi_{j-1}) \neq \emptyset$, we have $b \in \supp(\pi_{j-1})$. Then, by \eqref{eq:ineff-nonneg}, and the fact that $\pi_{j-1}(\cv)=\cv$,
\[
\ief {\cv} b {\pi_{j-1}}\ge\varphi(\cv,b)+\varphi(b,\cv)-\varphi(\cv,\cv)=2\varphi(\cv,b)\ge2\min_{v_{i}\in \supp(\pi_{j-1})}  \varphi(\cv,v_{i}).
\] 

\paragraph{$\pi_j$ is unbalanced and $\cv\not\in\supp(\pi_j)$.}
In this case, we must show 
\begin{equation}\label{eq:we7}
\frac12 \ief* {\tau_j} {\pi_{j-1}} \ge \min_{v_{i}\in \supp(\pi_{j-1})}  \varphi(\cv,v_{i}) -  \min_{v_{i}\in \supp(\pi_{j})}  \varphi(\cv,v_{i}).
\end{equation} 
Let $\tau_j=(a\,b)$. If $a,b\in\supp(\pi_{j-1})$, then $\supp(\pi_j)\subseteq\supp(\pi_{j-1})$ and thus 
\[
\min_{v_{i}\in \supp(\pi_{j-1})}  \varphi(\cv,v_{i}) \le \min_{v_{i}\in \supp(\pi_{j})}  \varphi(\cv,v_{i}).
\]
Hence the right side of \eqref{eq:we7} is non-positive and its left side is non-negative, so it holds.

Since $\supp(\tau_j) \cap \supp(\pi_{j-1}) \neq \emptyset$, we cannot
have both $a\not\in\supp(\pi_{j-1})$ and $b\not\in\supp(\pi_{j-1})$.
So as the final case, we may assume $a\not\in\supp(\pi_{j-1})$ but
$b\in\supp(\pi_{j-1})$. We may also assume that $\phi(a, v_c) < \phi(v_i, v_c)$
for all $v_i \in \supp(\pi_{j-1})$, since otherwise the right side of \eqref{eq:we7}
is again nonpositive, as $\supp(\pi_j) \subset \supp(\pi_{j-1}) \cup \{a\}$. Since $\pi_{j-1}(a) = a$, applying \eqref{eq:ineff-nonneg} yields
\begin{align*}
  \frac{1}{2}\ief*{\tau_j}{\pi_{j-1}} &\geq \frac{1}{2}\left(\varphi(a,b) + \varphi(b, \pi_{j-1}(a)) - \varphi(a, \pi_{j-1}(a))\right) \\
  &= \varphi(a,b) \\
  &\geq \varphi(b,v_c) - \varphi(a, v_c) \\
  &= \varphi(b,v_c) - \min_{v_i \in \supp(\pi_{j})} \varphi(v_i, v_c) \\
  &\geq \min_{v_i \in \supp(\pi_{j-1})} \varphi(v_i, v_c) - \min_{v_i \in \supp(\pi_{j})} \varphi(v_i, v_c).
\end{align*}
Note that the second inequality follows from the triangle inequality applied to $a$, $b$, and $v_c$.

This completes the proof of the fact that $f_j$ is non-decreasing, and
the proof of \eqref{eq:whatever5}.  We now show if $\pi$ has only one
(non-trivial) cycle, the lower bound of \eqref{eq:whatever5} is
achievable.

Consider the cycle $\pi=(v_1\,\dotsm\,v_{|\pi|})$ and let $v_j$ be the element of $\supp(\pi)$ that minimizes $\varphi(\cv,v_i)$. There are two cases to consider: either $v_{j}$ lies on the same branch as $v_{j+1}$, or it lies on a different branch (See Fig.~\ref{fig:10}.). If it lies on a 
different branch, we let $\pi=\pi_1(v_j\,\cv)$, where 
\begin{equation}\label{eq:whatever3}
\pi_1 = (v_1\,\dotsm\,v_{j}\,\cv\,v_{j+1}\,\dotsm\,v_{|\pi|}).
\end{equation}
Since $\pi_1$ contains the center,
$
\diste (\pi_{1})=\frac{1}{2}\Disp (\pi_{1})=\frac{1}{2}\Disp (\pi).
$
Hence, $\diste (\pi)\le\frac{1}{2}\Disp (\pi_{1})+\varphi(\cv,v_{j})$.

If $v_j$ lies on the same branch as $v_{j+1}$, let
\[
k = j - 1 + \min\{h:\pi^h(v_j)\not\in\branch(v_j),\}
\]
so that $v_{k+1}$ is the closest vertex following $v_j$ in the cycle that does not lie on the same branch as $v_j$. We then write $\pi = \pi_1\pi_2,$ where
\begin{equation}\begin{split}\label{eq:whatever4}
\pi_1 &= (v_1\,\dotsm\,v_j\,v_{k+1}\,\dotsm\,v_{|\pi|}),\\
\pi_2 &= (v_j\,\dotsm\,v_k).
\end{split}\end{equation}
Note that the cycle $\pi_1$ is unbalanced, but $v_j$ and $v_{k+1}$ lie on different branches. Hence, based on the analysis of the previous case, one has
\[
\diste (\pi_1)\le\frac12\Disp (\pi_1)+\varphi(\cv,v_j).
\]
As the support of $\pi_2$ is contained in a single branch, Lemma~\ref{lm:efficiency} implies that
\[
\diste (\pi_2)\le\frac12\Disp (\pi_2).
\]
Since $v_j$ and $v_k$ lie on the same branch and $\phi(v_j, v_c) \leq \phi(v_k, v_c)$, we see that $v_j$
lies on the $v_k,v_c$-path and, hence, the $v_k,v_{k+1}$ path. By Lemma~\ref{lem:onpath}, we have $D(\pi) = D(\pi_1) + D(\pi_2)$,
so that
\begin{align*}
\diste (\pi) &\le \diste (\pi_1) + \diste (\pi_2)\\
&= \frac12\Disp (\pi_1)+\varphi(\cv,v_j)+\frac12\Disp (\pi_2)\\
&= \frac12 \Disp (\pi)+\varphi(\cv,v_j).
\end{align*}
This completes the proof of the lemma.
\myqed\end{proof}

%A closed form expression is presented in Lemma \ref{lem:unbalanced}.
\begin{algorithm}
\SetKwFunction{concat}{Concatenate}
\SetKwFunction{this}{unbalanced-td}
\SetKwFunction{path}{path-td}
\SetKwFunction{central}{central-td}
\SetKwFunction{concat}{Concatenate}
\caption{unbalanced-td}
\label{alg:utd}
\KwIn{A cycle $\pi=(v_1\,v_2\,\cdots\, v_{|\pi|})$}
\KwOut{A minimum cost decomposition of $\pi$}
\BlankLine
{$v_{j}\gets\min_{v_{i}\in \supp(\pi)}\varphi{(\cv,v_i)}$}\;
\eIf{$v_j$ and $v_{j+1}$ lie on different branches}
{\Return \concat(\central($(v_1\,\dotsm\,v_{j}\,\cv\,v_{j+1}\,\dotsm\,v_{|\pi|})$), $(v_j\,\cv)$)\;}
{\Return \concat(\this($(v_1\,\dotsm\,v_j\,v_{k+1}\,\dotsm\,v_{|\pi|})$), \path($(v_j\,\dotsm\,v_k)$))\;}
\end{algorithm}

\begin{figure}
\centering
\subfigure[]{\includegraphics[width=4cm]{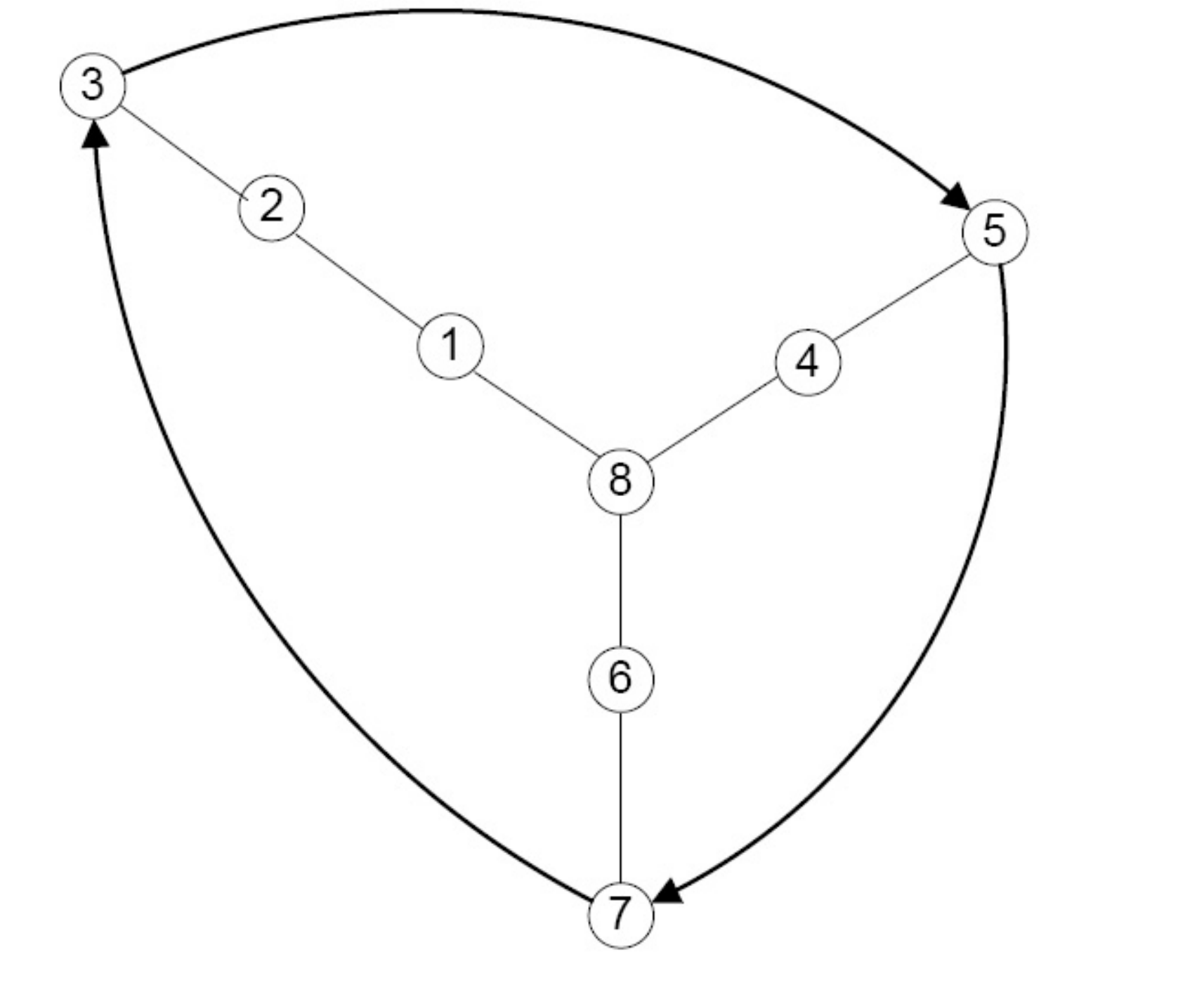}}
\subfigure[]{\includegraphics[width=4cm]{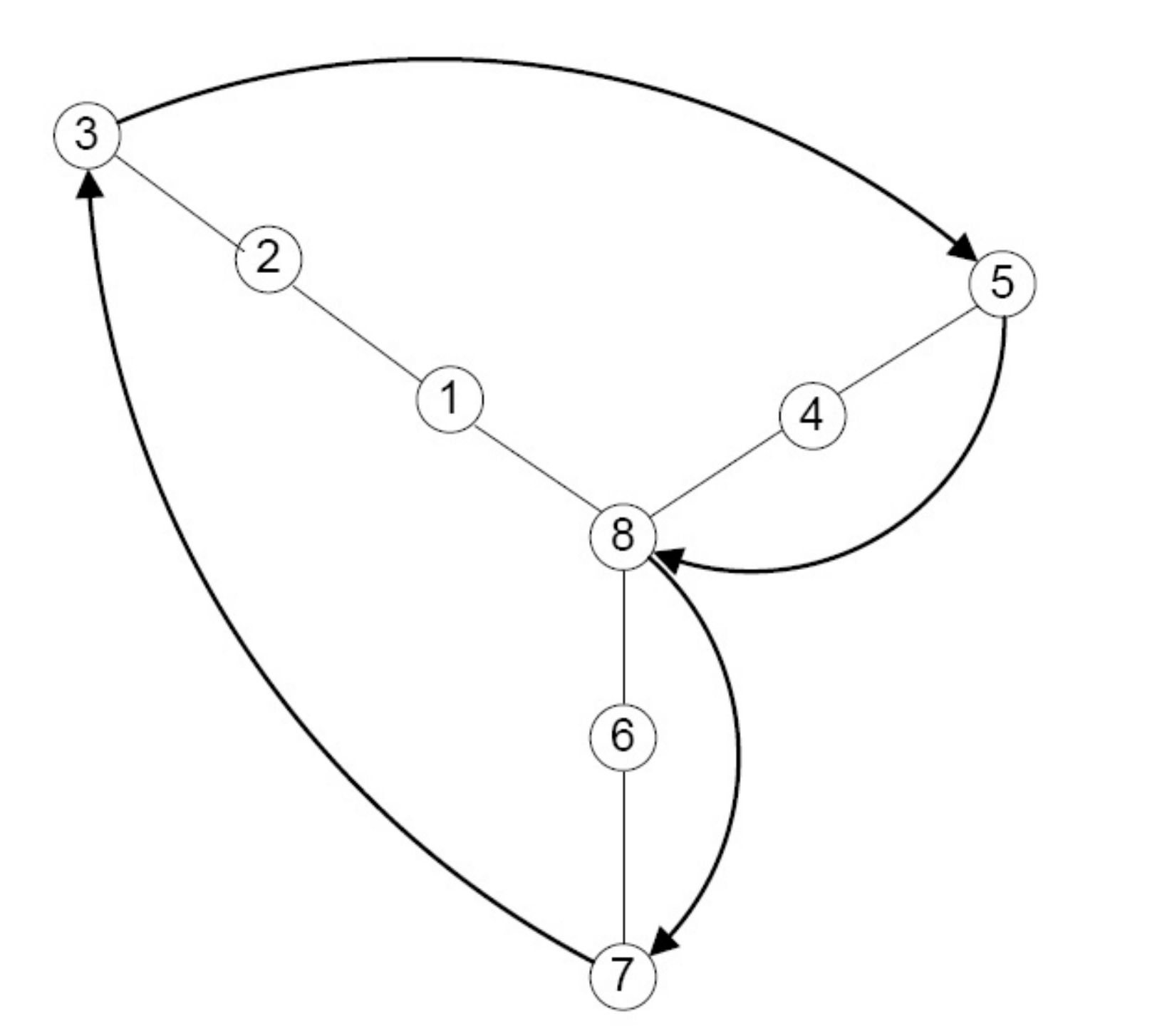}}
\subfigure[]{\includegraphics[width=4cm]{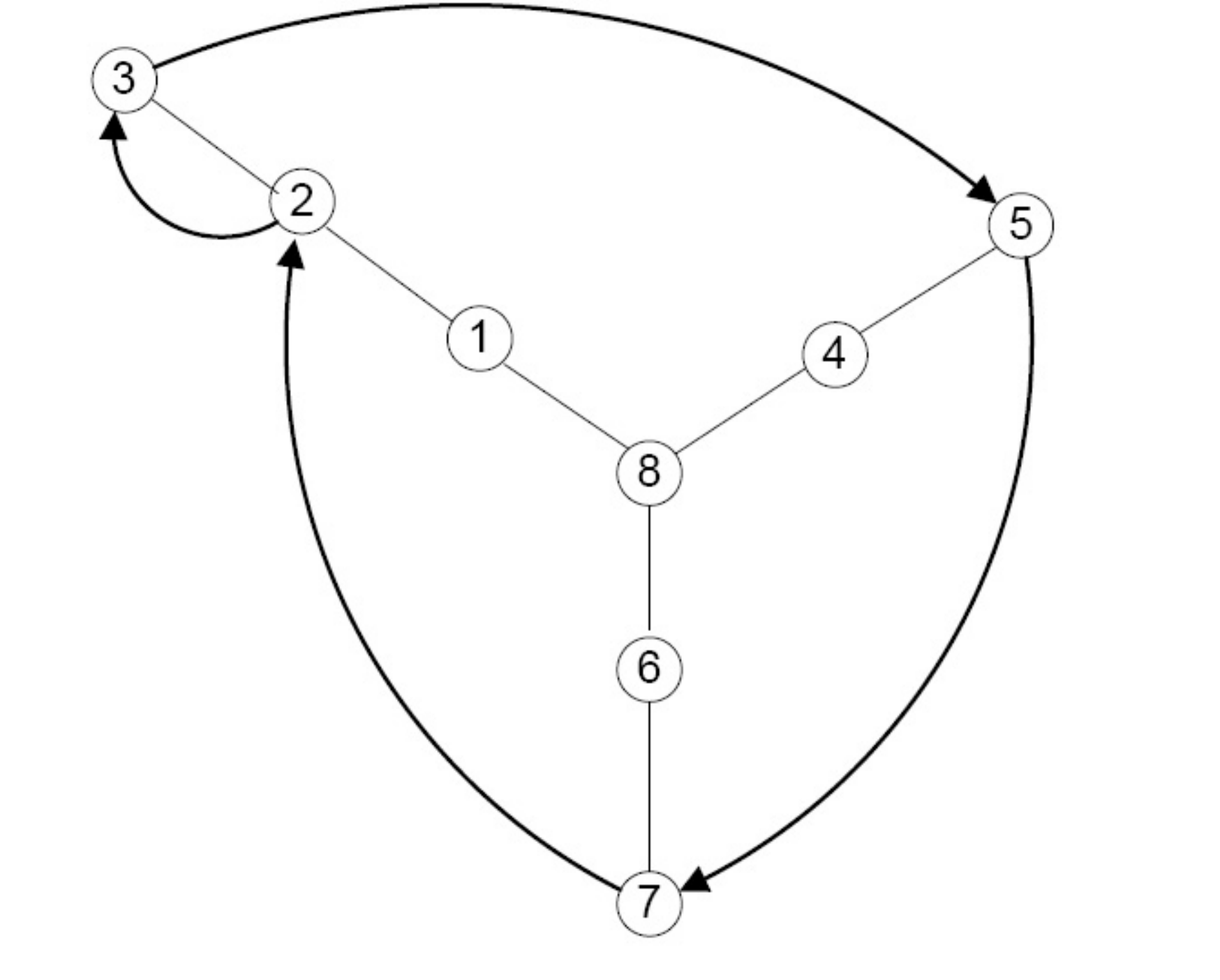}}
\subfigure[]{\includegraphics[width=4cm]{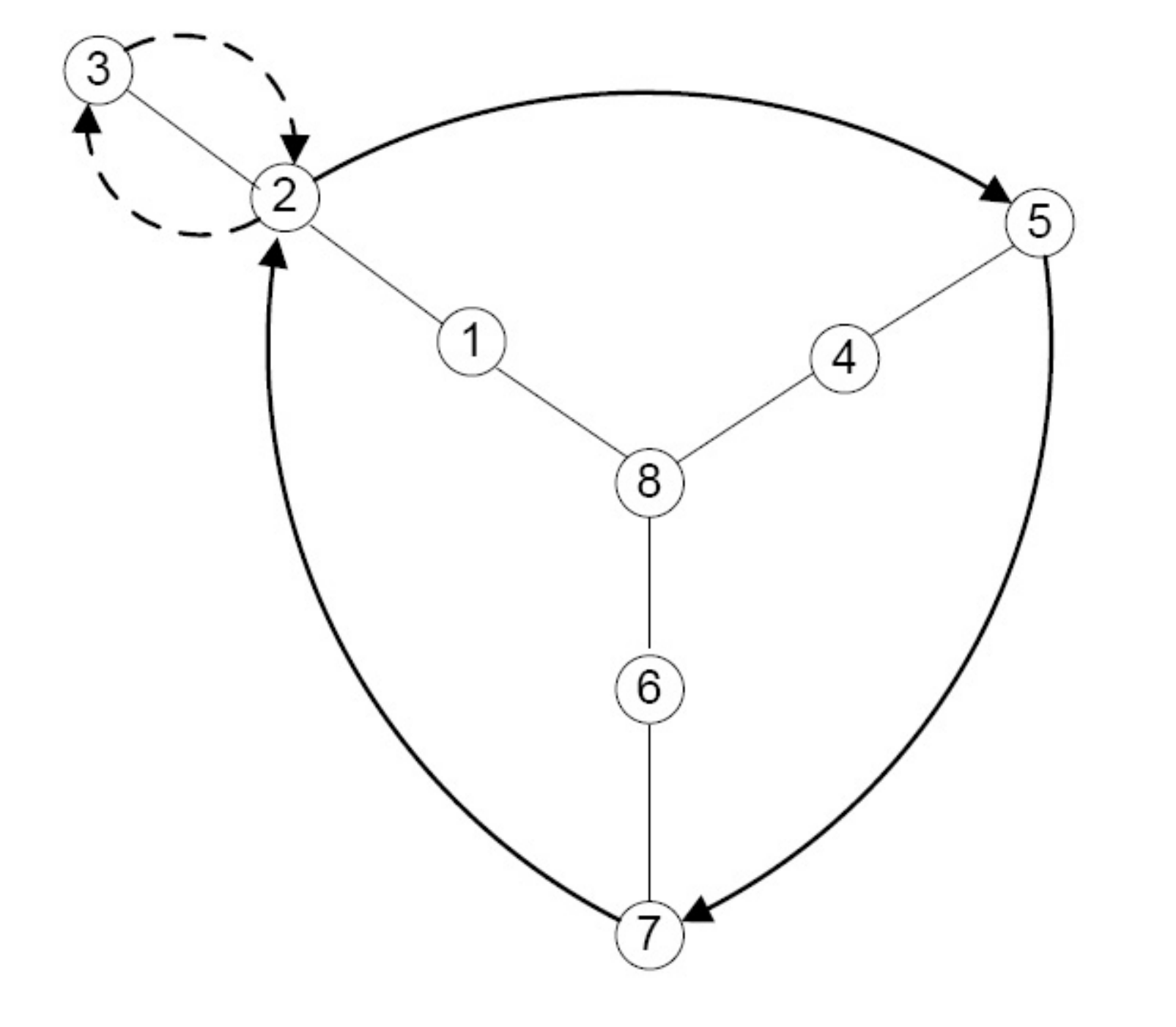}}

\caption{1\,) The cycle $\pi=(3\,5\,7)$. We have $v_j=5$, and $v_j,$ $v_{j+1}$ lie on different branches; 2\,) The cycle $\pi_1 = (3\,5\,8\,7)$ (See Equation~\eqref{eq:whatever3});\hspace{.2cm}
3\,) The cycle $\pi=(2\,3\,5\,7)$. We have $v_j=2$, and $v_j,$ $v_{j+1}$ lie on the same branch; 4\,) The cycles $\pi_1=(2\,5\,7)$ and $\pi_2=(2\,3)$ (See Equation~\eqref{eq:whatever4}).}
\label{fig:10}
\end{figure}

As a final remark, note that the exposition in Subsection \ref{case1}, Subsection \ref{case2} and Subsection \ref{case3} implicitly assumes that certain properties (such as balancedness) of a specific cycle are known beforehand. However, if this is not the case, additional steps have to be performed to test for such properties, and they reduces to straightforward search and counting procedures. The complexity of this search is linear in the size of the permutation.

We summarize our findings in the following theorem.
\begin{t1}\label{thm:cycles}
Let $\varphi$ be a Y-tree weight function and let $\pi$ be a cycle permutation. If $\pi$ does not contain the central vertex and is unbalanced, then
\[
\diste (\pi)=\frac {1}{2}\Disp (\pi)+ \min_{v_{i}\in \supp(\pi)}  \varphi(\cv,v_{i}).
\]
Otherwise,
\[
\diste (\pi)=\frac {1}{2}\Disp (\pi).
\]
\end{t1}

\begin{figure}[t]
\centering
\includegraphics[width=5cm]{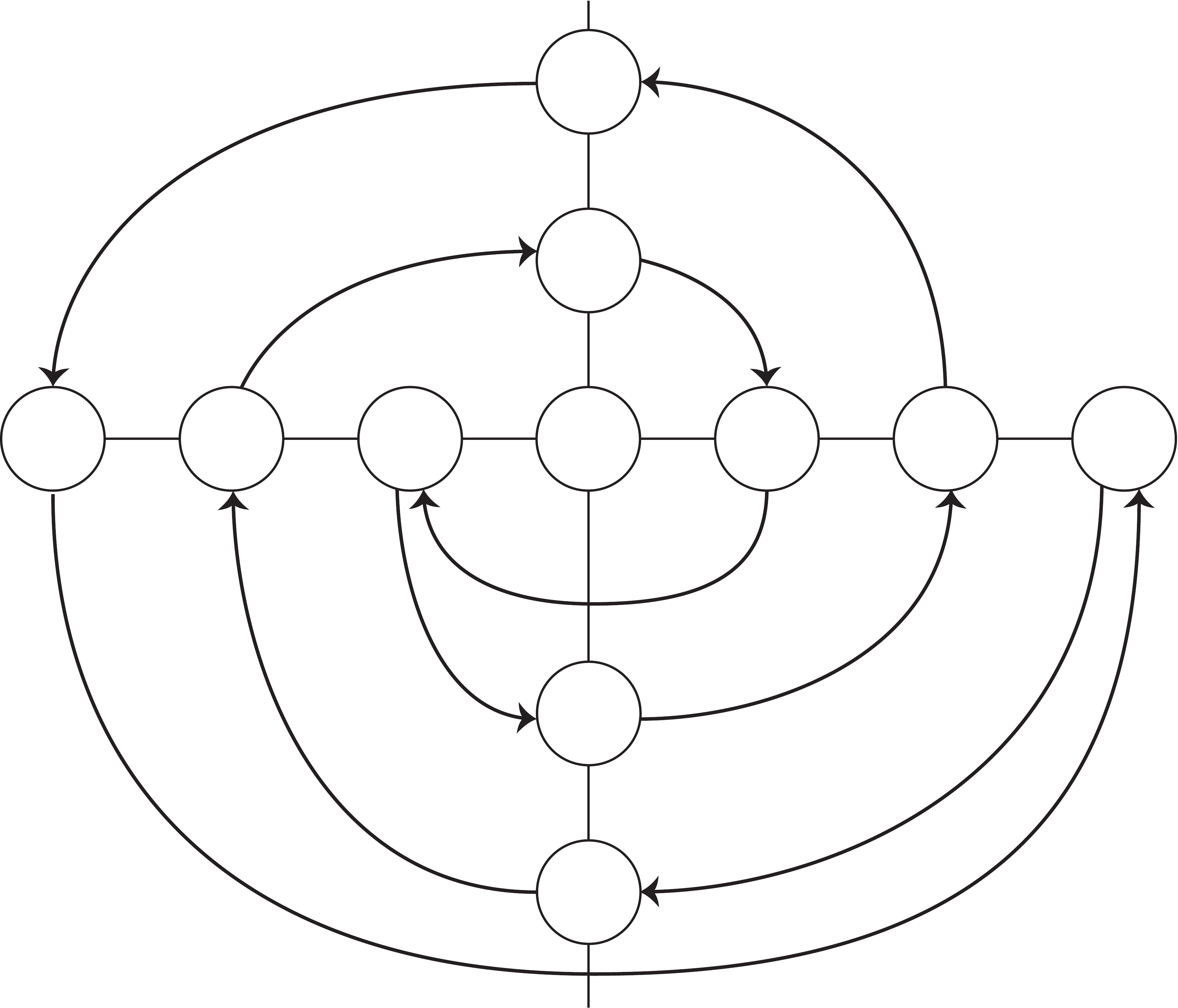}
\caption{Example of a balanced cycle on a star tree with four branches that cannot be optimally decomposed using Algorithm~\ref{alg:btd}. The labels of the vertices have no bearing on the finding, and are hence not included.}
\label{fig:counter_example}
\end{figure}

We conclude this section by noting that it may appear straightforward to extend the results of Algorithms~\ref{alg:ctd},~\ref{alg:btd}, and~\ref{alg:utd} to a more general defining tree model. This, unfortunately, is not the case even when one shifts from Y-trees, in which the unique node with degree larger than two has degree three, to so called \emph{star-trees}, in which the unique node with degree larger than two may have degree larger than three. In particular, Algorithm~\ref{alg:btd} cannot be immediately extended as it relies on the fact that a balanced cycle $\pi=(v_1\dotsc v_m)$ on a defining Y-tree, there exist two distinct indices $t$ and $l$ modulo $m$, such that $v_{t}$ and $v_{l+1}$ belong to the same branch, and $v_{t+1},\dotsc,v_{l}$ belong to a different branch (See proof of Lemma~\ref{lem:central}). An example of a balanced cycle on a star-tree that does not satisfy this property is shown in Fig.~\ref{fig:counter_example}.

\subsection{Computational Complexity of Decomposing Individual Cycles}

Careful examination of the algorithms described in the previous sections reveals that three major computational steps are involved in finding a minimum cost decomposition, including: 1\,) Identifying the type of the cycle; 2\,) conducting an adjacent cycle decomposition; 3\,) solving the individual sub-cycle decomposition problems with supports on paths. From a complexity viewpoint, step 1\,) requires $O(n)$ operations for checking whether the central vertex $\cv$ belongs to the cycle or not. If the central vertex belongs to the cycle, the decomposition calls for Algorithm~\ref{alg:ctd}, which requires $O(n)$ operations. Otherwise, in order to check whether the given cycle is balanced or unbalanced, one has to traverse the cycle to count the number of edges crossing branches and store/compare the values of $l_{ij}$ for all pairs of $i,j \in\{{1,2,3\}}$. This counting procedure requires $O(n)$ operations.

When Algorithm~\ref{alg:ctd} and Algorithm~\ref{alg:btd} are used, we follow the given cycle and at each vertex we check whether the optimal decomposition conditions are met. Each check requires constant computational time, and as Algorithm~\ref{alg:ctd} and Algorithm~\ref{alg:btd} terminate when each vertex in the cycle is visited exactly once and when the path decomposition if performed. To solve multiple path cycle decompositions individually, inductive arguments show that at most $m-2$ operations are needed, where $m$ denotes the length of the cycle. In addition, $\sum_{i=1}^{k}|\supp(\pi_i)|=|\pi|+(k-1)$, where $k$ is the number of cycles supported on paths in an adjacent cycle decomposition of $\pi$. As a result, since the complexity of both focal steps in the algorithm equals $O(n)$, the overall complexity of the methods equals $O(n)$.

Algorithm~\ref{alg:utd} proceeds along the same lines as Algorithm~\ref{alg:ctd}, except for an additional minimization procedure, which requires $O(n)$ operations. As a result, the complexity of this algorithm also equals $O(n)$.

\section{General Permutations}
\label{sec:GP}

Computing the weighted transposition distance between permutations with multiple cycles under the Y-tree weights model is significantly more challenging than computing the same distance between the identity and a single cycle. We currently do not know of any efficient procedure for computing this distance exactly for an arbitrary permutation. Nevertheless, in this section, we describe a straightforward linear-time $4/3$-approximation algorithm.%, the outline of which is given in Algorithm~\ref{alg11}.

Let us start by recalling a solution to the decomposition problem when all transposition weights are equal: perform the disjoint cycle decomposition and then sort each cycle independently. However, this independent cycle decomposition strategy does not always produce optimal solutions for general weight functions, as illustrated by the example of the permutation $\pi=(4, 6, 2, 5, 1, 3, 7)$ depicted in Fig.~\ref{mb}. Sorting each cycle of this permutation independently has total cost strictly larger than $\frac12\Disp (\pi)$. Alternatively, $\pi$ may be sorted by first applying the transposition $(1\,3)$, thereby merging the cycles $(1\,4\,5)$ and $(2\,6\,3)$. As the resulting cycle is balanced, it can be subsequently sorted via a sequence of efficient transpositions. Since the transposition $(1\,3)$ is efficient as well, the resulting transform has cost $\diste (\pi)=\frac{1}{2}\Disp (\pi)$. However, even the method of merging cycles may not always be optimal, as may be seen from the example given in Fig.~\ref{mi}.

\begin{figure}
\centering
\mbox{
\subfigure[]{\includegraphics[width=5cm]{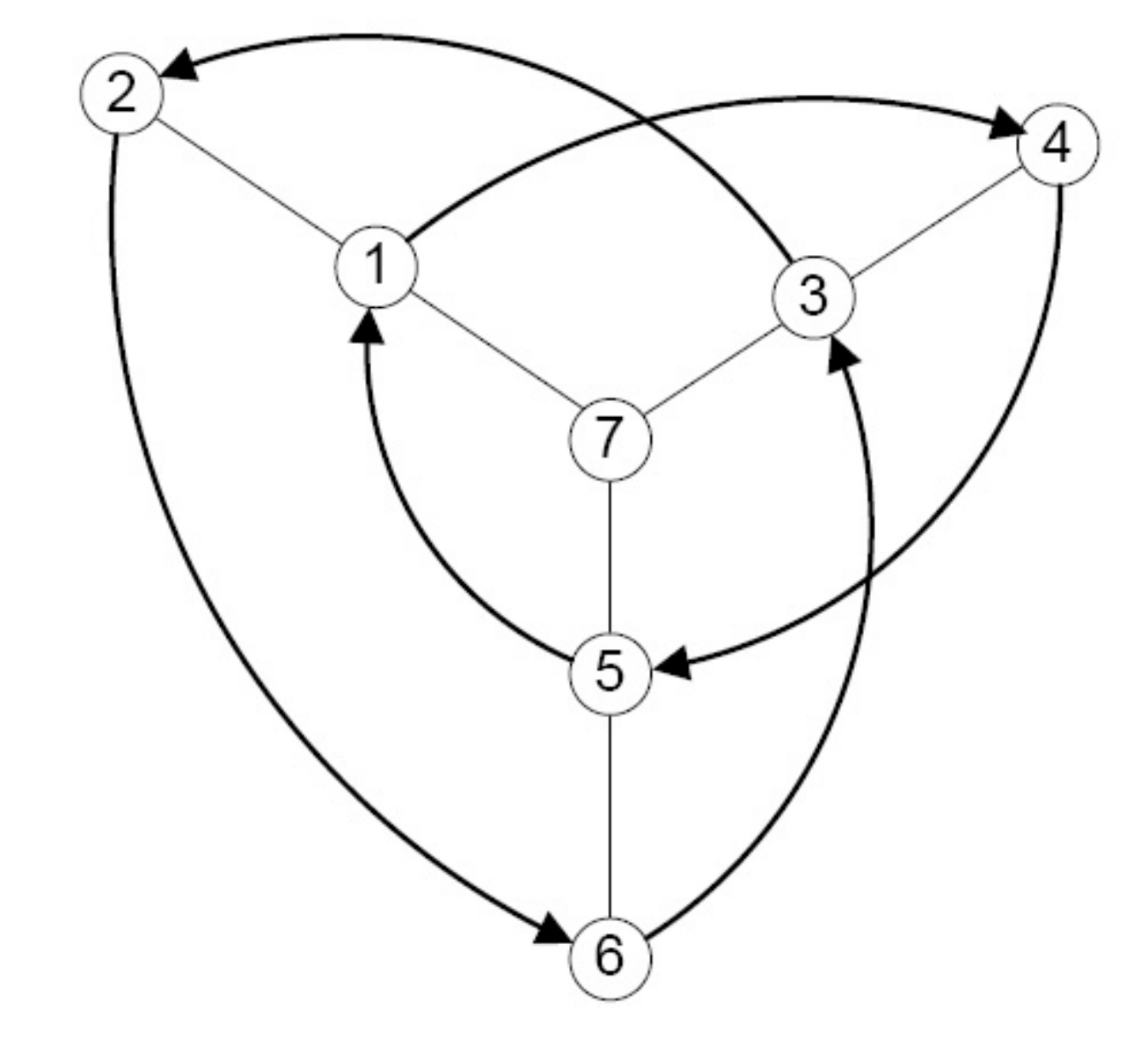}}\quad
\subfigure[]{\includegraphics[width=5cm]{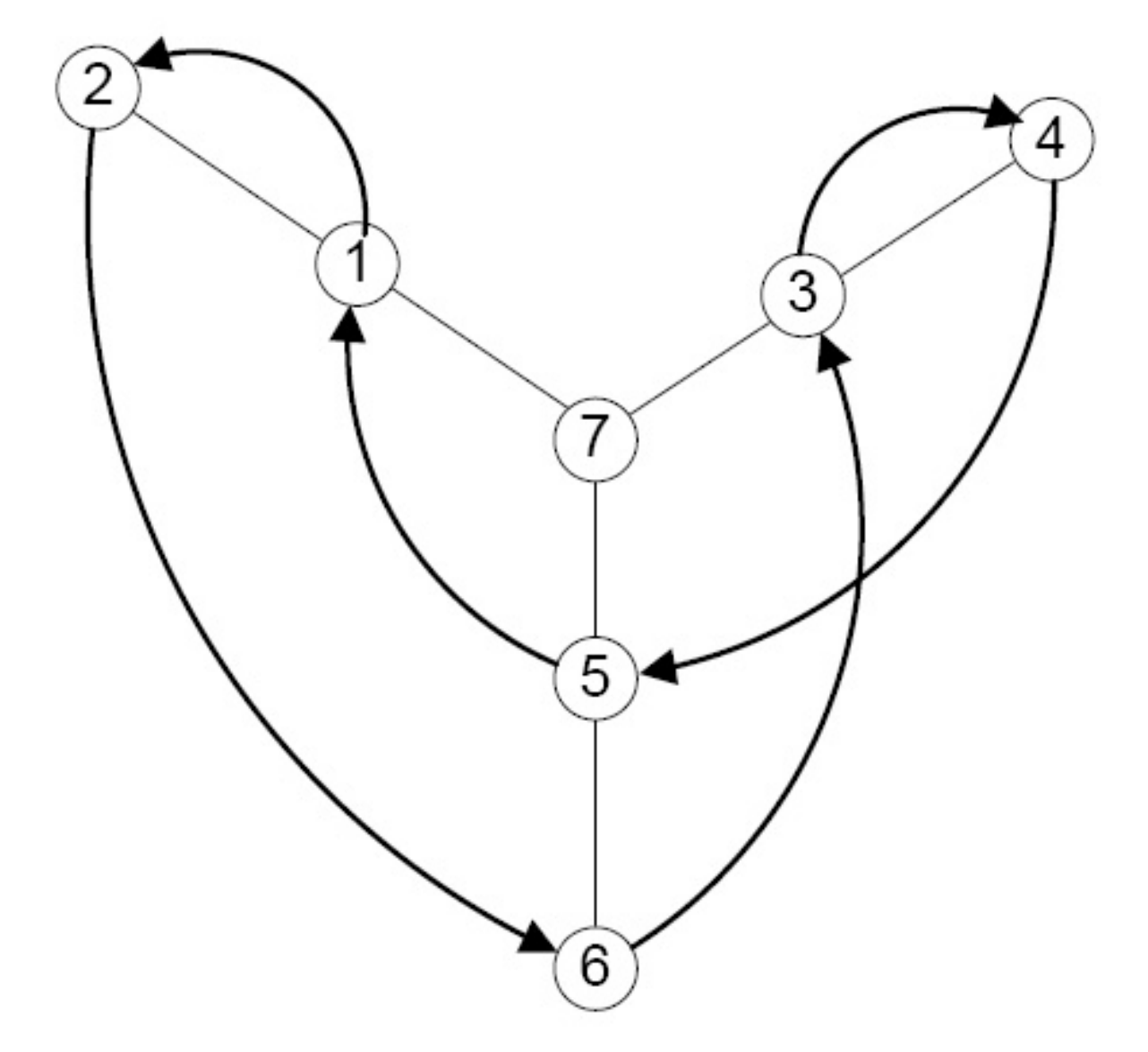}}}
\caption{Merging two cycles creates a balanced cycle: In 1\,) $\pi$ is shown as the product of two cycles (1\,4\,5) and (2\,6\,3); the merged cycle after applying transposition (1\,3) is shown in 2\,).}
\label{mb}
\end{figure}
\begin{figure}
\centering
\includegraphics[width=6.5cm]{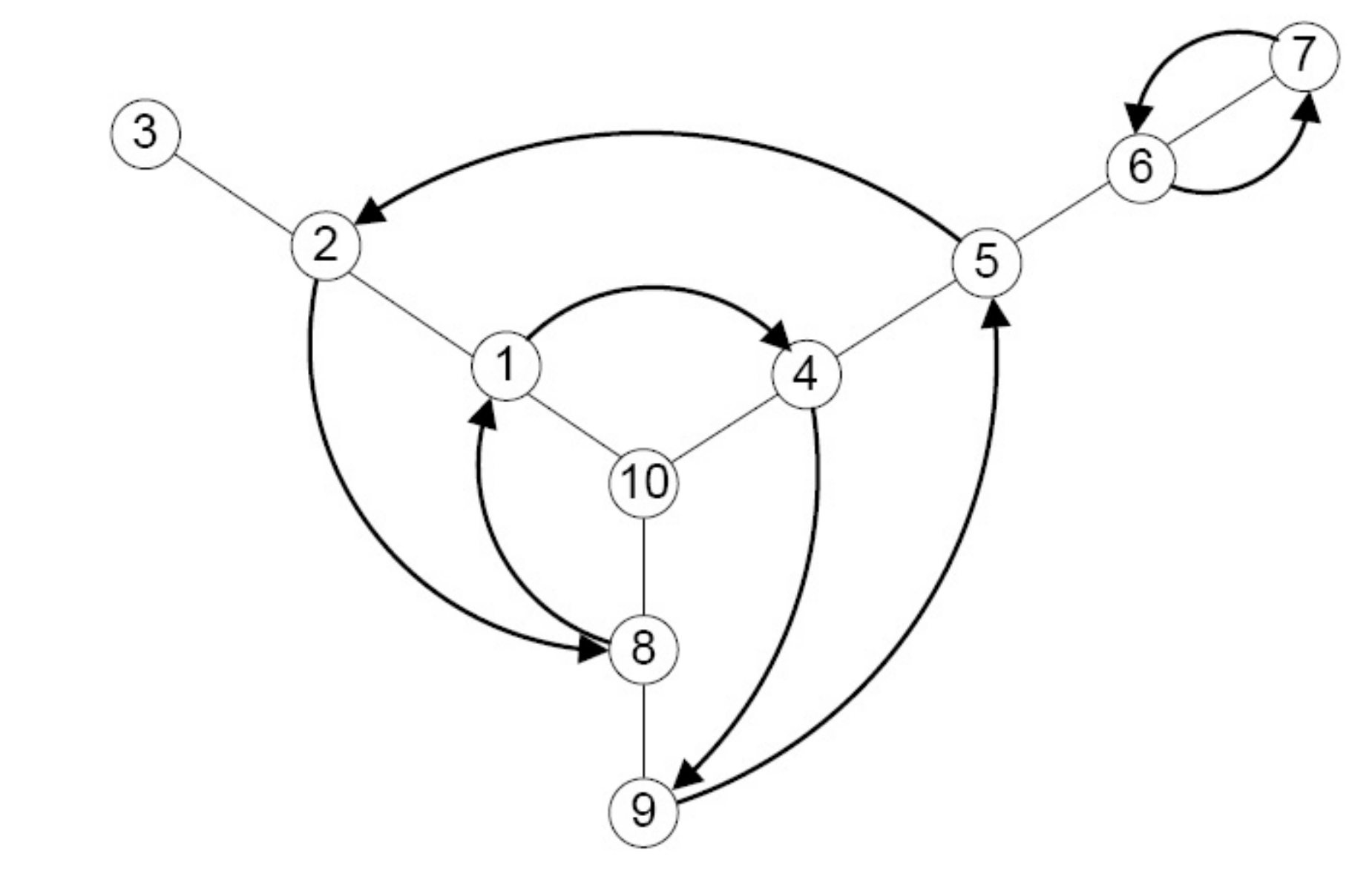}
\caption{An example illustrating that merging two cycles may lead to a suboptimal solution: For the 
permutation $\pi=(4, 8, 3, 9, 2, 7, 6, 1, 5, 10)$, via exhaustive search it can be determined that the minimum decomposition cost equals $\frac{1}{2}D(\pi,e)$ instead of $\frac{1}{2}D(\pi,e)+\varphi{(5,6)}$, which may be obtained via merging cycles. }
\label{mi}
\end{figure}

While decomposing every cycle independently may be in general sub-optimal, the process still provides a $4/3$-approximation to the optimal solution. To see this, we first prove that for any cycle $\kappa$, 
\begin{equation}\label{eq:whatever6}
\diste (\kappa)\le\frac{2}{3}\Disp (\kappa).
\end{equation}

For cycles that lie on a path of the Y-tree, cycles that contain the central vertex, and balanced cycles, this follows from Lemmas~\ref{lem:metric-path},~\ref{lem:central}, and~\ref{lem:balanced}, respectively. For an unbalanced cycle $\kappa$, from Lemma~\ref{lem:unbalanced}, we have
\[
\diste (\kappa)\le\frac{1}{2}\Disp (\kappa)+ \min_{v_{i}\in \supp(\kappa)}  \varphi(\cv,v_{i}).
\]
Hence, to show that $\diste (\kappa)\le\frac{2}{3}\Disp (\kappa)$, it suffices to prove that for an unbalanced cycle $\kappa$, $\min_{v_{i}\in \supp(\kappa)}  \varphi(\cv,v_{i})\le\frac16\Disp (\kappa)$. Let $w_1$, $w_2$, and $w_3,$ given by
\begin{align*}
w_{1} & =\min_{v_{i}\in \supp(\kappa)\cap B_{1}}\varphi(\cv, v_{i})\\
w_{2} & =\min_{v_{i}\in \supp(\kappa)\cap B_{2}}\varphi(\cv, v_{i})\\
w_{3} & =\min_{v_{i}\in \supp(\kappa)\cap B_{3}}\varphi(\cv, v_{i}),
\end{align*}
be the cost of transposing $\cv$ with the closest element to $\cv$ in $\supp(\kappa)$ on each of the three branches. Without loss of generality, assume that $w_1 =\min_{v_{i}\in \supp(\kappa)}  \varphi(\cv,v_{i})$, that is, $w_1\le w_2$ and $w_1\le w_2$. Since $\kappa$ is unbalanced, it must contain arcs between any two pair of branches. Thus, since there are at least three arcs,
\begin{align*}
\Disp (\kappa) & \ge\left(w_{1}+w_{2}\right)+\left(w_{2}+w_{3}\right)+\left(w_{3}+w_{1}\right)\\
 & \ge2\left(w_{1}+w_{2}+w_{3}\right)\\
 & \ge6w_{1}\\
 & = 6\min_{v_{i}\in\supp(\kappa)}\varphi(\cv, v_{i}),
\end{align*}
 which established the desired result. So~\eqref{eq:whatever6} holds for any cycle $\kappa$.

Let $\pi$ be a permutation with cycles $\kappa_1\dotsm\kappa_m$. If we decompose each cycle $\kappa_i$ independently of the other cycles using Algorithms~\ref{alg:ptd},~\ref{alg:ctd},~\ref{alg:btd}, and~\ref{alg:utd}, the total cost equals $\sum_{i=1}^m\delta(\kappa_i)$. This leads to
\[
\sum_{i=1}^m\delta(\kappa_i)\le\frac{2}{3}\sum_{i=1}^m\Disp (\kappa)=\frac23\Disp (\pi)\le\frac43\diste (\pi),
\]
where the first inequality follows from~\eqref{eq:whatever6} and the second inequality follows form Lemma~\ref{lem:metric-path}. Hence, the approximation factor is $4/3$, as claimed.

As a final remark, we would like to point out that the unbalanced cycles may be merged according to their lengths in order to provide  practical improvements to the theoretical approximation bound of $4/3$. The procedure asks for merging the central vertex $\cv$ with the unbalanced cycle of longest length, $m$, by using a vertex in its support closest to $\cv$. Given that there are $m$ arcs,
\begin{align*}
\Disp (\kappa) & \ge (2m) \min_{v_{i}\in\supp(\kappa)}\varphi(\cv, v_{i}).
\end{align*}
Once the central vertex $\cv$ is included in the newly formed cycle, one can merge other unbalanced cycles into the cycle via smallest cost transpositions involving the central vertex. In this case, the approximation constant equals $1+1/m$.

\section{Conclusion}

We introduced the notion of similarity distance between rankings under Y-tree weights and presented a polynomial-time algorithm for computing the distance between cycle permutations in terms of the displacement function. The algorithm was centered around the idea of adjacent cycle decomposition, i.e., rewriting a cycle as a product of adjacent/disjoint shorter cycles, where the support of each cycle can be embedded on a path in the defining graph of the Y-tree.

We also described a linear-time decomposition algorithm for permutations that may be embedded in the Y-tree as non-intersecting cycles, and the procedure reduced to finding the shortest path between two non-intersecting cycles. As for general permutations, we developed a linear time, $4/3$-approximation algorithm which is governed by the fact that if there exists a arc emanating from the central vertex that intersects all cycles across branches, then all cycles across branches can be merged efficiently.

\bibliographystyle{amsplain}
\bibliography{BibYtreeCopy}

\end{document}